\documentclass[submission,copyright,creativecommons]{eptcs}
\providecommand{\event}{QPL 2014} 
\usepackage{breakurl}             

\usepackage{graphicx}
\usepackage[fleqn]{amsmath}
\usepackage[varg]{txfonts}

\usepackage{color}

\usepackage{tikzfig}

\tikzstyle{env}=[copoint,regular polygon rotate=0,minimum width=0.2cm, fill=black]

\tikzstyle{every picture}=[baseline=-0.25em]
\tikzstyle{dotpic}=[scale=0.5]
\tikzstyle{diredges}=[every to/.style={diredge}]
\tikzstyle{dot graph}=[shorten <=-0.1mm,shorten >=-0.1mm,scale=0.6]
\tikzstyle{plot point}=[circle,fill=black,minimum width=2mm,inner sep=0]

\tikzstyle{braceedge}=[decorate,decoration={brace,amplitude=2mm,raise=-1mm}]
\tikzstyle{small braceedge}=[decorate,decoration={brace,amplitude=1mm,raise=-1mm}]
\tikzstyle{left hook arrow}=[left hook-latex]
\tikzstyle{right hook arrow}=[right hook-latex]

\tikzstyle{black dot}=[inner sep=0.7mm,minimum width=0pt,minimum height=0pt,fill=black,draw=black,shape=circle]

\tikzstyle{dot}=[black dot]
\tikzstyle{smalldot}=[inner sep=0.4mm,minimum width=0pt,minimum height=0pt,fill=black,draw=black,shape=circle]
\tikzstyle{white dot}=[dot,fill=white]
\tikzstyle{antipode}=[white dot,inner sep=0.3mm,font=\footnotesize]
\tikzstyle{smallwhitedot}=[smalldot,fill=white]
\tikzstyle{alt white dot}=[white dot,label={[xshift=3.07mm,yshift=-0.05mm,font=\footnotesize]left:$*$}]
\tikzstyle{gray dot}=[dot,fill=gray!40!white]
\tikzstyle{smallgraydot}=[smalldot,fill=gray!40!white]
\tikzstyle{box vertex}=[draw=black,rectangle]
\tikzstyle{small box}=[box vertex,fill=white]
\tikzstyle{whitebg}=[fill=white,inner sep=2pt]
\tikzstyle{graph state vertex}=[sg vertex,fill=black]

\tikzstyle{wide copoint}=[fill=white,draw=black,shape=isosceles triangle,shape border rotate=90,isosceles triangle stretches=true,inner sep=1pt,minimum width=1.5cm,minimum height=5mm]
\tikzstyle{wide point}=[fill=white,draw=black,shape=isosceles triangle,shape border rotate=-90,isosceles triangle stretches=true,inner sep=1pt,minimum width=1.5cm,minimum height=4mm]
\tikzstyle{very wide copoint}=[fill=white,draw=black,shape=isosceles triangle,shape border rotate=-90,isosceles triangle stretches=true,inner sep=1pt,minimum width=2.5cm,minimum height=4mm]
\tikzstyle{very wide empty copoint}=[draw=black,shape=isosceles triangle,shape border rotate=-90,isosceles triangle stretches=true,inner sep=1pt,minimum width=2.5cm,minimum height=4mm]
\tikzstyle{symm}=[ultra thick,shorten <=-1mm,shorten >=-1mm]

\tikzstyle{square box}=[rectangle,fill=white,draw=black,minimum height=5mm,minimum width=5mm,font=\small]
\tikzstyle{square gray box}=[rectangle,fill=gray!30,draw=black,minimum height=6mm,minimum width=6mm]
\tikzstyle{copoint}=[regular polygon,regular polygon sides=3,draw=black,scale=0.75,inner sep=-0.5pt,minimum width=7mm,fill=white]
\tikzstyle{point}=[regular polygon,regular polygon sides=3,draw=black,scale=0.75,inner sep=-0.5pt,minimum width=7mm,fill=white,regular polygon rotate=180]
\tikzstyle{gray point}=[point,fill=gray!40!white]
\tikzstyle{gray copoint}=[copoint,fill=gray!40!white]

\newcommand{\edgearrow}{{\arrow[black]{>}}}
\newcommand{\edgetick}{{\arrow[black,scale=0.7,very thick]{|}}}

\tikzstyle{diredge}=[->]
\tikzstyle{rdiredge}=[<-]
\tikzstyle{medium diredge}=[->]

\tikzstyle{short diredge}=[->]
\tikzstyle{halfedge}=[-)]
\tikzstyle{other halfedge}=[(-]
\tikzstyle{freeedge}=[(-)]
\tikzstyle{white edge}=[line width=5pt,white]
\tikzstyle{tick}=[postaction=decorate,decoration={markings, mark=at position 0.5 with \edgetick}]
\tikzstyle{small map edge}=[|-latex, gray!60!blue, shorten <=0.9mm, shorten >=0.5mm]
\tikzstyle{thick dashed edge}=[very thick,dashed,gray!40]
\tikzstyle{map edge}=[|-latex,very thick, gray!40, shorten <=1mm, shorten >=0.5mm]
\tikzstyle{tickedge}=[postaction=decorate,
  decoration={markings, mark=at position 0.5 with \edgetick}]
\tikzstyle{dirtickedge}=[postaction=decorate,
  decoration={markings, mark=at position 0.5 with \edgetick},
  decoration={markings, mark=at position 0.85 with \edgearrow}]
\tikzstyle{dirdoubletickedge}=[postaction=decorate,
  decoration={markings, mark=at position 0.4 with \edgetick},
  decoration={markings, mark=at position 0.6 with \edgetick},
  decoration={markings, mark=at position 0.85 with \edgearrow}]

\makeatletter
\newcommand{\boxshape}[3]{%
\pgfdeclareshape{#1}{
\inheritsavedanchors[from=rectangle] 
\inheritanchorborder[from=rectangle]
\inheritanchor[from=rectangle]{center}
\inheritanchor[from=rectangle]{north}
\inheritanchor[from=rectangle]{south}
\inheritanchor[from=rectangle]{west}
\inheritanchor[from=rectangle]{east}
\backgroundpath{
\southwest \pgf@xa=\pgf@x \pgf@ya=\pgf@y
\northeast \pgf@xb=\pgf@x \pgf@yb=\pgf@y

\@tempdima=#2
\@tempdimb=#3

\pgfpathmoveto{\pgfpoint{\pgf@xa - 5pt + \@tempdima}{\pgf@ya}}
\pgfpathlineto{\pgfpoint{\pgf@xa - 5pt - \@tempdima}{\pgf@yb}}
\pgfpathlineto{\pgfpoint{\pgf@xb + 5pt + \@tempdimb}{\pgf@yb}}
\pgfpathlineto{\pgfpoint{\pgf@xb + 5pt - \@tempdimb}{\pgf@ya}}
\pgfpathlineto{\pgfpoint{\pgf@xa - 5pt + \@tempdima}{\pgf@ya}}
\pgfpathclose
}
}}

\boxshape{NEbox}{0pt}{8pt}
\boxshape{SEbox}{0pt}{-8pt}
\boxshape{NWbox}{8pt}{0pt}
\boxshape{SWbox}{-8pt}{0pt}
\makeatother

\tikzstyle{map}=[draw,shape=NEbox,inner sep=7pt]
\tikzstyle{mapdag}=[draw,shape=SEbox,inner sep=7pt]
\tikzstyle{maptrans}=[draw,shape=SWbox,inner sep=7pt]
\tikzstyle{mapconj}=[draw,shape=NWbox,inner sep=7pt]

\tikzstyle{probs}=[shape=semicircle,fill=gray!40!white,draw=black,shape border rotate=180,minimum width=1.2cm]

\tikzstyle{arrs}=[-latex,font=\small,auto]
\tikzstyle{arrow plain}=[arrs]
\tikzstyle{arrow dashed}=[dashed,arrs]
\tikzstyle{arrow bold}=[very thick,arrs]
\tikzstyle{arrow hide}=[draw=white!0,-]
\tikzstyle{arrow reverse}=[latex-]
\tikzstyle{cdnode}=[]

\tikzstyle{gn}=[dot,fill=green,minimum width=0.3cm,inner sep=0pt]
\tikzstyle{rn}=[dot,fill=red,inner sep=0pt,minimum width=0.3cm]
\tikzstyle{bn}=[dot,fill=blue,minimum width=0.3cm]

\tikzstyle{rc}=[dot,thick,fill=white,draw = red,minimum width=0.3cm,inner sep=0pt]
\tikzstyle{gc}=[dot,thick,fill=white,draw= green,inner sep=0pt,minimum width=0.3cm]
\tikzstyle{bc}=[dot,thick,fill=white,draw= blue,minimum width=0.3cm]

\tikzstyle{label}=[circle,fill=white,minimum width=0.3cm]

\tikzstyle{H box}=[rectangle,draw=black,xscale=1,yscale=1,font=\small,inner sep=0.75pt]
\tikzstyle{clocklabel}=[dot,fill=yellow,draw=black,font=\tiny,inner sep=0.75pt]

\tikzstyle{rsn}=[circle split,draw,fill=red,font=\tiny,inner sep=0.75pt]
\tikzstyle{gsn}=[circle split,draw,fill=green,font=\tiny,inner sep=0.75pt]
\tikzstyle{bsn}=[circle split,draw,fill=blue,font=\tiny,inner sep=0.75pt]

\tikzstyle{rsc}=[circle split,thick,draw= red,draw,fill=white,font=\tiny,inner sep=0.75pt]
\tikzstyle{gsc}=[circle split,thick,draw= green,draw,fill=white,font=\tiny,inner sep=0.75pt]
\tikzstyle{bsc}=[circle split,thick,draw= blue,draw,fill=white,font=\tiny,inner sep=0.75pt]

\tikzstyle{cnot}=[fill=white,shape=circle,inner sep=-1.4pt]
\tikzstyle{wire label}=[font=\tiny, auto]

\newcommand{\bra}[1]{\ensuremath{\left\langle #1 \right|}}
\newcommand{\ket}[1]{\ensuremath{\left|  #1 \right\rangle}}

\tikzstyle{cdiag}=[matrix of math nodes, row sep=3em, column sep=3em, text height=1.5ex, text depth=0.25ex,inner sep=0.5em]
\tikzstyle{arrow above}=[transform canvas={yshift=0.5ex}]
\tikzstyle{arrow below}=[transform canvas={yshift=-0.5ex}]

\newtheorem{Th}{Theorem}[section]
\newtheorem{theorem}[Th]{Theorem}
\newtheorem{proposition}[Th]{Proposition}
\newtheorem{lemma}[Th]{Lemma}

\newenvironment{proof}{\textbf{Proof:}}{\hfill$\Box$\newline}

\title{Qutrit Dichromatic Calculus and Its Universality}
\author{
Quanlong Wang \qquad\qquad Xiaoning Bian
\institute{School of Mathematics and Systems Science\\
Beihang University\\
Beijing, China}
\email{\quad qlwang@buaa.edu.cn \quad\qquad bianxiaoning@smss.buaa.edu.cn}
}
\def\titlerunning{Qutrit Dichromatic Calculus and Its Universality}
\def\authorrunning{Q. Wang \& X. Bian}
\begin{document}
\maketitle

\begin{abstract}
We introduce a dichromatic calculus (RG) for qutrit systems. We show that the decomposition of the qutrit Hadamard gate is non-unique and not derivable from the dichromatic calculus. As an application of the dichromatic calculus, we depict a quantum algorithm with a single qutrit. Since it is not easy to decompose an arbitrary $d\times d$ unitary matrix into Z and X phase gates when $d>2$, the proof of the universality of qudit ZX calculus for quantum mechanics is far from trivial. We construct a counterexample to Ranchin's universality proof, and give another proof by Lie theory that the qudit ZX calculus contains all single qudit unitary transformations, which implies that qudit ZX calculus, with qutrit dichromatic calculus as a special case, is universal for quantum mechanics.

\end{abstract}

\section{Introduction}

In \cite{CoeckeDuncan}, Coecke and Duncan developed dichromatic ZX-calculus for qubit systems. To extend the graphical calculus to higher dimensions, Ranchin considered qudit (d-dimensional quantum system) ZX-calculus \cite{Ranchin}. At almost the same time, the authors of this paper investigated the theory and application of qutrit ZX-calculus \cite{wangbian}. Unlike in \cite{Ranchin} and \cite{wangbian}, we introduce two new rules P1 and P2 in this paper. The necessity of these two rules is demonstrated by
depicting in dichromatic calculus the simplest quantum speed-up algorithm with a single qutrit \cite{Gedik}.

In the qubit case, Duncan and Perdrix \cite{RossPerdrix} proved that the Euler decomposition is not derivable from ZX calculus. In this paper, we  prove similarly that the decomposition of the qutrit Hadamard gate is non-unique and not derivable from a dichromatic qutrit ZX-calculus.

For any d-dimensional quantum system ($d\geq2$), universality is a very important problem for ZX calculus. This means that the qudit ZX calculus can express any quantum state and gate. To the best of our knowledge, it is not easy to decompose an arbitrary $d\times d$ unitary matrix into Z and X phase gates when $d>2$. Thus the proof of the universality of qudit ZX calculus for quantum mechanics is far from trivial. Due to Brylinski \cite{Bry}, to prove the universality of qudit ZX calculus, it suffices to prove that the qudit ZX calculus contains all single qudit unitary transformations. Such a proof given in \cite{Ranchin} is based on the fact \cite{Muth} that the d-dimensional phase gates $Z_d$ and  $X_d$ are sufficient to simulate all single qudit unitary transforms. For our understanding, only part of the whole family of $Z_d$ phase gates can be represented by $\Lambda_X$   phase gates (i.e., X phase gates)
 in \cite{Ranchin}. Actually, we have a counterexample that some $Z_d$ phase gates cannot be realized by $\Lambda_X$  only. Thus another proof that the qudit ZX calculus contains all single qudit unitary transformations is requested. We solve this problem by the method of Lie algebra. Therefore the qudit ZX calculus, with qutrit dichromatic calculus as a special case, is universal for quantum mechanics.

\section{Red and Green Graphs}
We fix some notations here. Let $\mathbf{FdHilb}$ be the symmetric
monoidal $\dag-$category(SM$\dag-$category) of
finite-dimensional complex Hilbert spaces and linear
maps between them. Let $\mathbf{FdHilb_p}$ be The SM$\dag-$category of
finite-dimensional complex Hilbert spaces and linear
maps modulo the relation $f\equiv g$ if $\exists z\in \mathbb{C},z\neq0:f=zg$. $\mathbf{FdHilb_Q}$ is defined as the full subcategory
of $\mathbf{FdHilb_p}$ generated by the objects $\{\underbrace{Q\otimes\cdots\otimes Q}_n\mid n\geq0$, where $Q:=\mathbb{C}^3$. This is
essentially the category of qutrits.
\subsection{$\mathbf{RG}$ category}
We define a category $\mathbf{RG}$ where the objects are $n$-fold
monoidal products of an object $*$, denoted $*^n (n\geq0)$.
In $\mathbf{RG}$, a morphism from $*^m$ to $*^n$ is a finite undirected open graph from $m$ wires to $n$ wires, built from

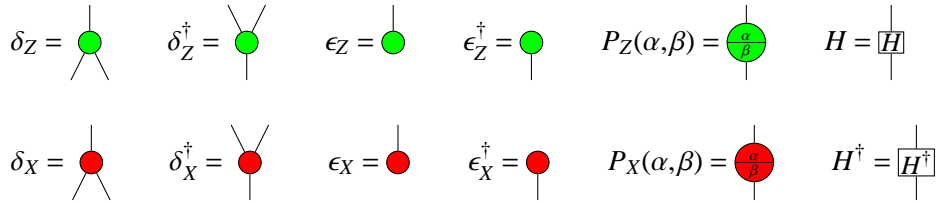
\begin{center}
\[
\delta_Z=%
\beginpgfgraphicnamed{RGgenerator/RGg_dz}
\begin{tikzpicture}
	\begin{pgfonlayer}{nodelayer}
		\node [style=none] (0) at (0, 0.5) {};
		\node [style=none] (1) at (0.25, -0.5) {};
		\node [style=none] (2) at (-0.25, -0.5) {};
		\node [style=gn] (3) at (0, -0) {};
	\end{pgfonlayer}
	\begin{pgfonlayer}{edgelayer}
		\draw (0.center) to (3);
		\draw (3) to (2.center);
		\draw (3) to (1.center);
	\end{pgfonlayer}
\end{tikzpicture}}
\endpgfgraphicnamed \qquad
\delta_Z^\dagger=%
\beginpgfgraphicnamed{RGgenerator/RGg_dzd}
\begin{tikzpicture}
	\begin{pgfonlayer}{nodelayer}
		\node [style=gn] (0) at (0, -0) {};
		\node [style=none] (1) at (0, -0.5) {};
		\node [style=none] (2) at (0.25, 0.5) {};
		\node [style=none] (3) at (-0.25, 0.5) {};
	\end{pgfonlayer}
	\begin{pgfonlayer}{edgelayer}
		\draw (0) to (1.center);
		\draw (0) to (2.center);
		\draw (0) to (3.center);
	\end{pgfonlayer}
\end{tikzpicture}}
\endpgfgraphicnamed \qquad
\epsilon_Z=%
\beginpgfgraphicnamed{RGgenerator/RGg_ez}
\begin{tikzpicture}
	\begin{pgfonlayer}{nodelayer}
		\node [style=none] (0) at (0, 0.5) {};
		\node [style=gn] (1) at (0, -0) {};
	\end{pgfonlayer}
	\begin{pgfonlayer}{edgelayer}
		\draw (0.center) to (1);
	\end{pgfonlayer}
\end{tikzpicture}}
\endpgfgraphicnamed \qquad
\epsilon_Z^\dagger=%
\beginpgfgraphicnamed{RGgenerator/RGg_ezd}
\begin{tikzpicture}
	\begin{pgfonlayer}{nodelayer}
		\node [style=none] (0) at (0, -0.5) {};
		\node [style=gn] (1) at (0, -0) {};
	\end{pgfonlayer}
	\begin{pgfonlayer}{edgelayer}
		\draw (0.center) to (1);
	\end{pgfonlayer}
\end{tikzpicture}}
\endpgfgraphicnamed \qquad
P_Z(\alpha,\beta)=%
\beginpgfgraphicnamed{RGgenerator/RGg_zph_ab}
\begin{tikzpicture}
	\begin{pgfonlayer}{nodelayer}
		\node [style=gsn] (0) at (0, -0) {$\alpha$\nodepart{lower}$\beta$};
		\node [style=none] (1) at (0, 0.5) {};
		\node [style=none] (2) at (0, -0.5) {};
	\end{pgfonlayer}
	\begin{pgfonlayer}{edgelayer}
		\draw (1.center) to (0);
		\draw (2.center) to (0);
	\end{pgfonlayer}
\end{tikzpicture}}
\endpgfgraphicnamed \qquad
H = %
\beginpgfgraphicnamed{RGgenerator/RGg_Hada}
\begin{tikzpicture}
	\begin{pgfonlayer}{nodelayer}
		\node [style={H box}] (0) at (0, -0) {$H$};
		\node [style=none] (1) at (0, 0.5) {};
		\node [style=none] (2) at (0, -0.5) {};
	\end{pgfonlayer}
	\begin{pgfonlayer}{edgelayer}
		\draw (1.center) to (0);
		\draw (2.center) to (0);
	\end{pgfonlayer}
\end{tikzpicture}}
\endpgfgraphicnamed
\]
\end{center}
\[
\delta_X=%
\beginpgfgraphicnamed{RGgenerator/RGg_dx}
\begin{tikzpicture}
	\begin{pgfonlayer}{nodelayer}
		\node [style=none] (0) at (0, 0.5) {};
		\node [style=none] (1) at (0.25, -0.5) {};
		\node [style=none] (2) at (-0.25, -0.5) {};
		\node [style=rn] (3) at (0, -0) {};
	\end{pgfonlayer}
	\begin{pgfonlayer}{edgelayer}
		\draw (0.center) to (3);
		\draw (3) to (2.center);
		\draw (3) to (1.center);
	\end{pgfonlayer}
\end{tikzpicture}}
\endpgfgraphicnamed \qquad
\delta_X^\dagger=%
\beginpgfgraphicnamed{RGgenerator/RGg_dxd}
\begin{tikzpicture}
	\begin{pgfonlayer}{nodelayer}
		\node [style=rn] (0) at (0, -0) {};
		\node [style=none] (1) at (0, -0.5) {};
		\node [style=none] (2) at (0.25, 0.5) {};
		\node [style=none] (3) at (-0.25, 0.5) {};
	\end{pgfonlayer}
	\begin{pgfonlayer}{edgelayer}
		\draw (0) to (1.center);
		\draw (0) to (2.center);
		\draw (0) to (3.center);
	\end{pgfonlayer}
\end{tikzpicture}}
\endpgfgraphicnamed \qquad
\epsilon_X=%
\beginpgfgraphicnamed{RGgenerator/RGg_ex}
\begin{tikzpicture}
	\begin{pgfonlayer}{nodelayer}
		\node [style=none] (0) at (0, 0.5) {};
		\node [style=rn] (1) at (0, -0) {};
	\end{pgfonlayer}
	\begin{pgfonlayer}{edgelayer}
		\draw (0.center) to (1);
	\end{pgfonlayer}
\end{tikzpicture}}
\endpgfgraphicnamed \qquad
\epsilon_X^\dagger=%
\beginpgfgraphicnamed{RGgenerator/RGg_exd}
\begin{tikzpicture}
	\begin{pgfonlayer}{nodelayer}
		\node [style=none] (0) at (0, -0.5) {};
		\node [style=rn] (1) at (0, -0) {};
	\end{pgfonlayer}
	\begin{pgfonlayer}{edgelayer}
		\draw (0.center) to (1);
	\end{pgfonlayer}
\end{tikzpicture}}
\endpgfgraphicnamed \qquad
P_X(\alpha,\beta)=%
\beginpgfgraphicnamed{RGgenerator/RGg_xph_ab}
\begin{tikzpicture}
	\begin{pgfonlayer}{nodelayer}
		\node [style=rsn] (0) at (0, -0) {$\alpha$\nodepart{lower}$\beta$};
		\node [style=none] (1) at (0, 0.5) {};
		\node [style=none] (2) at (0, -0.5) {};
	\end{pgfonlayer}
	\begin{pgfonlayer}{edgelayer}
		\draw (1.center) to (0);
		\draw (2.center) to (0);
	\end{pgfonlayer}
\end{tikzpicture}}
\endpgfgraphicnamed \qquad
H^{\dag}= %
\beginpgfgraphicnamed{RGgenerator/RGg_Hadad}
\begin{tikzpicture}
	\begin{pgfonlayer}{nodelayer}
		\node [style={H box}] (0) at (0, -0) {$H^\dagger$};
		\node [style=none] (1) at (0, -0.5) {};
		\node [style=none] (2) at (0, 0.5) {};
	\end{pgfonlayer}
	\begin{pgfonlayer}{edgelayer}
		\draw (2.center) to (0);
		\draw (1.center) to (0);
	\end{pgfonlayer}
\end{tikzpicture}}
\endpgfgraphicnamed
\]


where $\alpha, \beta\in [0,2\pi)$. For convenience, we denote the
frequently used angles $\frac{2 \pi}{3}$ and   $\frac{4 \pi}{3}$ by $1$ and $2$ respectively.
The generator H is called a Hadamard gate.
Additionally, the identity morphism on $*$ is represented as the straight wire.
 Composition is connecting up the edges, while tensor is simply putting two diagrams side by
side. We also mention here that
we ignore connected components of a graph which
are connected to neither input nor output. This is in
order to not have to deal with scalars.

$\mathbf{RG}$ morphisms are also subject to the equations depicted
below.

\begin{figure}[!h]
\begin{center}
\[
\quad \qquad\begin{array}{|cccc|}
\hline
\beginpgfgraphicnamed{RGrelations/spider}
\begin{tikzpicture}[font={\footnotesize}]
	\begin{pgfonlayer}{nodelayer}
		\node [style=none] (0) at (-1.25, -0) {\rotatebox[origin=c]{45}{...}};
		\node [style=none] (1) at (0.25, -0) {$=$};
		\node [style=gsn] (2) at (1.5, 0) {\tiny $\alpha+\eta$\nodepart{lower}\tiny $\beta+\theta$};
		\node [style=gsn] (3) at (-0.75, -0.25) {$\eta$\nodepart{lower}$\theta$};
		\node [style=none] (4) at (-1.75, -0.5) {\raisebox{2mm}{...}};
		\node [style=none] (5) at (2, -0.75) {};
		\node [style=none] (6) at (-1, -0.75) {};
		\node [style=none] (7) at (1.5, -0.75) {\raisebox{2mm}{...}};
		\node [style=none] (8) at (-0.5, -0.75) {};
		\node [style=none] (9) at (1, -0.75) {};
		\node [style=none] (10) at (-2, -0.5) {};
		\node [style=none] (11) at (-1.5, -0.5) {};
		\node [style=none] (12) at (-0.75, -0.75) {\raisebox{2mm}{...}};
		\node [style=none] (13) at (1.5, 0.75) {\raisebox{-2mm}{...}};
		\node [style=none] (14) at (1, 0.75) {};
		\node [style=none] (15) at (-2, 0.75) {};
		\node [style=none] (16) at (-0.5, 0.5) {};
		\node [style=none] (17) at (-1.5, 0.75) {};
		\node [style=none] (18) at (2, 0.75) {};
		\node [style=gsn] (19) at (-1.75, 0.25) {$\alpha$\nodepart{lower}$\beta$};
		\node [style=none] (20) at (-0.75, 0.5) {\raisebox{-2mm}{...}};
		\node [style=none] (21) at (-1, 0.5) {};
		\node [style=none] (22) at (-1.75, 0.75) {\raisebox{-2mm}{...}};
	\end{pgfonlayer}
	\begin{pgfonlayer}{edgelayer}
		\draw (3) to (16.center);
		\draw (3) to (6.center);
		\draw (3) to (8.center);
		\draw (19) to (10.center);
		\draw (19) to (11.center);
		\draw [bend right, looseness=1.00] (19) to (3);
		\draw [bend left, looseness=1.00] (19) to (3);
		\draw (14.center) to (2);
		\draw (2) to (9.center);
		\draw (5.center) to (2);
		\draw (2) to (18.center);
		\draw (19) to (15.center);
		\draw (19) to (17.center);
		\draw (3) to (21.center);
	\end{pgfonlayer}
\end{tikzpicture}}
\endpgfgraphicnamed&(S1)&%
\beginpgfgraphicnamed{RGrelations/s2}
\begin{tikzpicture}
	\begin{pgfonlayer}{nodelayer}
		\node [style=gn] (0) at (-0.75, -0) {};
		\node [style=none] (1) at (0, -0) {$:=$};
		\node [style=none] (2) at (1, -0.5) {};
		\node [style=none] (3) at (1, 0.5) {};
		\node [style=none] (4) at (-0.75, -0.5) {};
		\node [style=none] (5) at (-0.75, 0.5) {};
		\node [style=gsn] (6) at (1, -0) {\tiny $0$\nodepart{lower}\tiny $0$};
		\node [style=none] (7) at (2.75, 0.5) {};
		\node [style=none] (8) at (2.75, -0.5) {};
		\node [style=none] (9) at (2, -0) {$=$};
	\end{pgfonlayer}
	\begin{pgfonlayer}{edgelayer}
		\draw (6) to (2.center);
		\draw (6) to (3.center);
		\draw (0) to (4.center);
		\draw (0) to (5.center);
		\draw (7.center) to (8.center);
	\end{pgfonlayer}
\end{tikzpicture}}
\endpgfgraphicnamed&(S2)\\
\beginpgfgraphicnamed{RGrelations/b1}
\begin{tikzpicture}
	\begin{pgfonlayer}{nodelayer}
		\node [style=none] (0) at (0.75, -0.25) {};
		\node [style=none] (1) at (-1, -0.5) {};
		\node [style=none] (2) at (-0.5, -0.5) {};
		\node [style=none] (3) at (0, -0) {$=$};
		\node [style=rn] (4) at (1.25, 0.25) {};
		\node [style=rn] (5) at (-0.75, 0.5) {};
		\node [style=gn] (6) at (-0.75, -0) {};
		\node [style=none] (7) at (1.25, -0.25) {};
		\node [style=rn] (8) at (0.75, 0.25) {};
	\end{pgfonlayer}
	\begin{pgfonlayer}{edgelayer}
		\draw [style=none] (5) to (6);
		\draw [style=none] (6) to (1.center);
		\draw [style=none] (6) to (2.center);
		\draw [style=none] (8) to (0.center);
		\draw [style=none] (4) to (7.center);
	\end{pgfonlayer}
\end{tikzpicture}}
\endpgfgraphicnamed&(B1)&%
\beginpgfgraphicnamed{RGrelations/b2}
\begin{tikzpicture}
	\begin{pgfonlayer}{nodelayer}
		\node [style=none] (0) at (-1.75, 1) {};
		\node [style=rn] (1) at (1.25, 0.25) {};
		\node [style=gn] (2) at (-1, 0.5) {};
		\node [style=none] (3) at (-1.75, -0.75) {};
		\node [style=none] (4) at (1.5, -0.75) {};
		\node [style=none] (5) at (1, 0.75) {};
		\node [style=none] (6) at (0, -0) {$=$};
		\node [style=none] (7) at (1, -0.75) {};
		\node [style=gn] (8) at (-1.75, 0.5) {};
		\node [style=rn] (9) at (-1.75, -0.25) {};
		\node [style=none] (10) at (-1, 1) {};
		\node [style=none] (11) at (-1, -0.75) {};
		\node [style=none] (12) at (1.5, 0.75) {};
		\node [style=gn] (13) at (1.25, -0.25) {};
		\node [style=rn] (14) at (-1, -0.25) {};
	\end{pgfonlayer}
	\begin{pgfonlayer}{edgelayer}
		\draw [style=none] (11.center) to (14);
		\draw [style=none] (3.center) to (9);
		\draw [style=none] (2) to (10.center);
		\draw [style=none, bend right, looseness=1.00] (14) to (2);
		\draw [style=none] (8) to (0.center);
		\draw [style=none, bend left, looseness=1.00] (9) to (8);
		\draw [style=none] (4.center) to (13);
		\draw [style=none] (13) to (1);
		\draw [style=none] (1) to (5.center);
		\draw [style=none] (1) to (12.center);
		\draw (13) to (7.center);
		\draw (8) to (14);
		\draw (2) to (9);
	\end{pgfonlayer}
\end{tikzpicture}}
\endpgfgraphicnamed&(B2)\\
\multicolumn{3}{|c}{%
\beginpgfgraphicnamed{RGrelations/k1}
\begin{tikzpicture}
	\begin{pgfonlayer}{nodelayer}
		\node [style=none] (0) at (0.75, 1) {};
		\node [style=rn] (1) at (-3.25, -0) {};
		\node [style=gsn] (2) at (2.5, -0) {\tiny $2$\nodepart{lower}\tiny $1$};
		\node [style=rn] (3) at (2.75, 0.5) {};
		\node [style=none] (4) at (0.5, -0.5) {};
		\node [style=none] (5) at (-1.5, -0.5) {};
		\node [style=none] (6) at (-2.5, 0.25) {$=$};
		\node [style=none] (7) at (1, -0.5) {};
		\node [style=gsn] (8) at (-1.5, -0) {\tiny $1$\nodepart{lower}\tiny $2$};
		\node [style=gsn] (9) at (0.75, 0.5) {\tiny $2$\nodepart{lower}\tiny $1$};
		\node [style=none] (10) at (3, -0.5) {};
		\node [style=none] (11) at (1.5, 0.25) {$=$};
		\node [style=gsn] (12) at (-1, -0) {\tiny $1$\nodepart{lower}\tiny $2$};
		\node [style=none] (13) at (-1, -0.5) {};
		\node [style=gsn] (14) at (-3.25, 0.5) {\tiny $1$\nodepart{lower}\tiny $2$};
		\node [style=gsn] (15) at (3, -0) {\tiny $2$\nodepart{lower}\tiny $1$};
		\node [style=none] (16) at (-3.25, 1) {};
		\node [style=rn] (17) at (0.75, -0) {};
		\node [style=none] (18) at (-1.25, 1) {};
		\node [style=none] (19) at (-3, -0.5) {};
		\node [style=none] (20) at (2.5, -0.5) {};
		\node [style=rn] (21) at (-1.25, 0.5) {};
		\node [style=none] (22) at (-3.5, -0.5) {};
		\node [style=none] (23) at (2.75, 1) {};
	\end{pgfonlayer}
	\begin{pgfonlayer}{edgelayer}
		\draw [style=none] (1) to (22.center);
		\draw [style=none] (1) to (19.center);
		\draw [style=none] (18.center) to (21);
		\draw [style=none] (17) to (4.center);
		\draw [style=none] (17) to (7.center);
		\draw [style=none] (23.center) to (3);
		\draw [style=none] (16.center) to (14);
		\draw [style=none] (14) to (1);
		\draw [style=none] (21) to (8);
		\draw [style=none] (8) to (5.center);
		\draw [style=none] (21) to (12);
		\draw [style=none] (12) to (13.center);
		\draw [style=none] (0.center) to (9);
		\draw [style=none] (9) to (17);
		\draw [style=none] (3) to (2);
		\draw [style=none] (2) to (20.center);
		\draw [style=none] (3) to (15);
		\draw [style=none] (15) to (10.center);
	\end{pgfonlayer}
\end{tikzpicture}}
\endpgfgraphicnamed}&(K1)\\
\multicolumn{3}{|c}{%
\beginpgfgraphicnamed{RGrelations/k2}
\begin{tikzpicture}
	\begin{pgfonlayer}{nodelayer}
		\node [style=gsn] (0) at (-4, 0.5) {\tiny $1$\nodepart{lower}\tiny $2$};
		\node [style=rsn] (1) at (-4, -0.25) { $\alpha$\nodepart{lower} $\beta$};
		\node [style=none] (2) at (-3.5, -0) {$=$};
		\node [style=none] (3) at (-4, 1) {};
		\node [style=none] (4) at (-4, -0.75) {};
		\node [style=none] (5) at (-3, -0.75) {};
		\node [style=none] (6) at (-3, 1) {};
		\node [style=gsn] (7) at (-3, -0.25) {\tiny $1$\nodepart{lower}\tiny $2$};
		\node [style=none] (8) at (-1.75, 1) {};
		\node [style=none] (9) at (-1.75, -0.75) {};
		\node [style=none] (10) at (-1.25, -0) {$=$};
		\node [style=gsn] (11) at (-0.75, -0.25) {\tiny $2$\nodepart{lower}\tiny $1$};
		\node [style=none] (12) at (-0.75, -0.75) {};
		\node [style=none] (13) at (-0.75, 1) {};
		\node [style=rsn] (14) at (-1.75, -0.25) { $\alpha$\nodepart{lower} $\beta$};
		\node [style=gsn] (15) at (-1.75, 0.5) {\tiny $2$\nodepart{lower}\tiny $1$};
		\node [style=rsn] (16) at (-3, 0.5) {\tiny $\beta\textnormal{-}\alpha$\nodepart{lower}\tiny $\textnormal{-}\alpha$};
		\node [style=rsn] (17) at (-0.75, 0.5) {\tiny $\textnormal{-}\beta$\nodepart{lower}\tiny $\alpha\textnormal{-}\beta$};
		\node [style=none] (18) at (3.75, -0.75) {};
		\node [style=none] (19) at (1, -0) {$=$};
		\node [style=none] (20) at (0.5, 1) {};
		\node [style=none] (21) at (2.75, 1) {};
		\node [style=none] (22) at (0.5, -0.75) {};
		\node [style=rsn] (23) at (1.5, -0.25) {\tiny $1$\nodepart{lower}\tiny $2$};
		\node [style=none] (24) at (3.25, -0) {$=$};
		\node [style=rsn] (25) at (2.75, 0.5) {\tiny $2$\nodepart{lower}\tiny $1$};
		\node [style=gsn] (26) at (0.5, -0.25) { $\alpha$\nodepart{lower} $\beta$};
		\node [style=none] (27) at (2.75, -0.75) {};
		\node [style=rsn] (28) at (0.5, 0.5) {\tiny $1$\nodepart{lower}\tiny $2$};
		\node [style=none] (29) at (1.5, -0.75) {};
		\node [style=gsn] (30) at (2.75, -0.25) { $\alpha$\nodepart{lower} $\beta$};
		\node [style=rsn] (31) at (3.75, -0.25) {\tiny $2$\nodepart{lower}\tiny $1$};
		\node [style=gsn] (32) at (3.75, 0.5) {\tiny $\beta\textnormal{-}\alpha$\nodepart{lower}\tiny $\textnormal{-}\alpha$};
		\node [style=none] (33) at (3.75, 1) {};
		\node [style=none] (34) at (1.5, 1) {};
		\node [style=gsn] (35) at (1.5, 0.5) {\tiny $\textnormal{-}\beta$\nodepart{lower}\tiny $\alpha\textnormal{-}\beta$};
	\end{pgfonlayer}
	\begin{pgfonlayer}{edgelayer}
		\draw (3.center) to (0);
		\draw (0) to (1);
		\draw (1) to (4.center);
		\draw (7) to (5.center);
		\draw (8.center) to (15);
		\draw (15) to (14);
		\draw (14) to (9.center);
		\draw (11) to (12.center);
		\draw (6.center) to (16);
		\draw (16) to (7);
		\draw (13.center) to (17);
		\draw (17) to (11);
		\draw (20.center) to (28);
		\draw (28) to (26);
		\draw (26) to (22.center);
		\draw (23) to (29.center);
		\draw (21.center) to (25);
		\draw (25) to (30);
		\draw (30) to (27.center);
		\draw (31) to (18.center);
		\draw (33.center) to (32);
		\draw (34.center) to (35);
		\draw (35) to (23);
		\draw (32) to (31);
	\end{pgfonlayer}
\end{tikzpicture}}
\endpgfgraphicnamed}&(K2)\\
\beginpgfgraphicnamed{RGrelations/h1}
\begin{tikzpicture}
	\begin{pgfonlayer}{nodelayer}
		\node [style={H box}] (0) at (0, 0.25) {$H^\dagger$};
		\node [style=none] (1) at (-1, 0.75) {};
		\node [style=none] (2) at (1, 0.5) {};
		\node [style=none] (3) at (0, 0.75) {};
		\node [style={H box}] (4) at (-1, -0.25) {$H^\dagger$};
		\node [style=none] (5) at (0.5, -0) {$=$};
		\node [style={H box}] (6) at (-1, 0.25) {$H$};
		\node [style=none] (7) at (0, -0.75) {};
		\node [style=none] (8) at (1, -0.5) {};
		\node [style=none] (9) at (-1, -0.75) {};
		\node [style=none] (10) at (-0.5, -0) {$=$};
		\node [style={H box}] (11) at (0, -0.25) {$H$};
	\end{pgfonlayer}
	\begin{pgfonlayer}{edgelayer}
		\draw (3.center) to (0);
		\draw (7.center) to (11);
		\draw (11) to (0);
		\draw (8.center) to (2.center);
		\draw (1.center) to (6);
		\draw (9.center) to (4);
		\draw (4) to (6);
	\end{pgfonlayer}
\end{tikzpicture}}
\endpgfgraphicnamed&(H1)&%
\beginpgfgraphicnamed{RGrelations/h2}
\begin{tikzpicture}
	\begin{pgfonlayer}{nodelayer}
		\node [style=none] (0) at (1, 0.75) {};
		\node [style=none] (1) at (-1.25, -1) {};
		\node [style={H box}] (2) at (-1.25, 0.5) {$H$};
		\node [style=none] (3) at (-1.75, -0.75) {\raisebox{2mm}{...}};
		\node [style=none] (4) at (1.5, -0.75) {\raisebox{2mm}{...}};
		\node [style=none] (5) at (1.5, 0.75) {\raisebox{-2mm}{...}};
		\node [style=rsn] (6) at (-1.75, -0) {$\alpha$\nodepart{lower}$\beta$};
		\node [style={H box}] (7) at (-2.25, -0.5) {$H^\dagger$};
		\node [style=none] (8) at (-1.25, 1) {};
		\node [style={H box}] (9) at (-1.25, -0.5) {$H^\dagger$};
		\node [style=none] (10) at (-1.75, 0.75) {\raisebox{-2mm}{...}};
		\node [style=none] (11) at (0, -0) {$=$};
		\node [style=gsn] (12) at (1.5, -0) {$\alpha$\nodepart{lower}$\beta$};
		\node [style=none] (13) at (2, -0.75) {};
		\node [style=none] (14) at (2, 0.75) {};
		\node [style=none] (15) at (-2.25, 1) {};
		\node [style={H box}] (16) at (-2.25, 0.5) {$H$};
		\node [style=none] (17) at (-2.25, -1) {};
		\node [style=none] (18) at (1, -0.75) {};
	\end{pgfonlayer}
	\begin{pgfonlayer}{edgelayer}
		\draw (6) to (7);
		\draw (6) to (9);
		\draw (9) to (1.center);
		\draw (7) to (17.center);
		\draw (6) to (16);
		\draw (6) to (2);
		\draw (2) to (8.center);
		\draw (16) to (15.center);
		\draw (12) to (0.center);
		\draw (12) to (14.center);
		\draw (12) to (18.center);
		\draw (12) to (13.center);
	\end{pgfonlayer}
\end{tikzpicture}}
\endpgfgraphicnamed&(H2)\\
\beginpgfgraphicnamed{RGrelations/p1s}
\begin{tikzpicture}
	\begin{pgfonlayer}{nodelayer}
		\node [style=none] (0) at (0.5, -0.5) {};
		\node [style={H box}] (1) at (-1.25, -0) {$D$};
		\node [style=none] (2) at (-0.5, 0.5) {};
		\node [style=rn] (3) at (0.25, 0.25) {};
		\node [style=gn] (4) at (-0.25, -0.25) {};
		\node [style=none] (5) at (0.75, -0) {$=$};
		\node [style=none] (6) at (-1.25, 0.5) {};
		\node [style=none] (7) at (-0.75, -0) {$:=$};
		\node [style=none] (8) at (-1.25, -0.5) {};
		\node [style=none] (9) at (1.25, 0.75) {};
		\node [style=gn] (10) at (1.25, 0.25) {};
		\node [style=rn] (11) at (1.25, -0.25) {};
		\node [style=none] (12) at (1.25, -0.75) {};
	\end{pgfonlayer}
	\begin{pgfonlayer}{edgelayer}
		\draw [bend left=15, looseness=1.00] (4) to (2.center);
		\draw (6.center) to (1);
		\draw (10) to (9.center);
		\draw (1) to (8.center);
		\draw (4) to (3);
		\draw [bend left=15, looseness=1.00] (3) to (0.center);
		\draw (11) to (12.center);
		\draw [bend right=45, looseness=1.00] (10) to (11);
		\draw [bend left=45, looseness=1.00] (10) to (11);
	\end{pgfonlayer}
\end{tikzpicture}}
\endpgfgraphicnamed&(P1)&%
\beginpgfgraphicnamed{RGrelations/p2s}
\begin{tikzpicture}[font={\footnotesize}]
	\begin{pgfonlayer}{nodelayer}
		\node [style={H box}] (0) at (-0.25, -0.5) {$D$};
		\node [style=none] (1) at (0.25, -0.5) {};
		\node [style=none] (2) at (-0.5, 1) {};
		\node [style=none] (3) at (-1.25, -0.5) {};
		\node [style={H box}] (4) at (-2.25, 0.25) {$D$};
		\node [style=gsn] (5) at (1.25, -0.25) {$\alpha$\nodepart{lower}$\beta$};
		\node [style=none] (6) at (1.75, -0) {$=$};
		\node [style=none] (7) at (0.5, 0.5) {};
		\node [style={H box}] (8) at (-0.75, -0.5) {$D$};
		\node [style=none] (9) at (-0.75, -1) {};
		\node [style=none] (10) at (-2.25, 0.75) {};
		\node [style=gn] (11) at (-0.5, -0) {};
		\node [style=none] (12) at (-1.75, -0) {$=$};
		\node [style=none] (13) at (0.75, -0.5) {};
		\node [style=none] (14) at (0, -0) {$=$};
		\node [style=none] (15) at (-1.25, 0.5) {};
		\node [style=none] (16) at (-2.25, -0.75) {};
		\node [style={H box}] (17) at (1.25, 0.25) {$D$};
		\node [style=none] (18) at (-0.25, -1) {};
		\node [style={H box}] (19) at (-0.5, 0.5) {$D$};
		\node [style=none] (20) at (1.25, 0.75) {};
		\node [style={H box}] (21) at (-2.25, -0.25) {$D$};
		\node [style=gn] (22) at (0.5, -0) {};
		\node [style=none] (23) at (1.25, -0.75) {};
		\node [style=none] (24) at (2.25, 0.75) {};
		\node [style=gsn] (25) at (2.25, 0.25) {$\beta$\nodepart{lower}$\alpha$};
		\node [style={H box}] (26) at (2.25, -0.25) {$D$};
		\node [style=none] (27) at (2.25, -0.75) {};
	\end{pgfonlayer}
	\begin{pgfonlayer}{edgelayer}
		\draw (3.center) to (15.center);
		\draw (11) to (19);
		\draw (19) to (2.center);
		\draw (11) to (8);
		\draw (8) to (9.center);
		\draw (11) to (0);
		\draw (0) to (18.center);
		\draw (22) to (7.center);
		\draw (22) to (1.center);
		\draw (22) to (13.center);
		\draw (5) to (17);
		\draw (17) to (20.center);
		\draw (10.center) to (4);
		\draw (16.center) to (21);
		\draw (21) to (4);
		\draw (5) to (23.center);
		\draw (25) to (26);
		\draw (26) to (27.center);
		\draw (25) to (24.center);
	\end{pgfonlayer}
\end{tikzpicture}}
\endpgfgraphicnamed&(P2)\\
\hline
\end{array}\]
\end{center}

  \caption{RG rules}\label{figure1}
\end{figure}
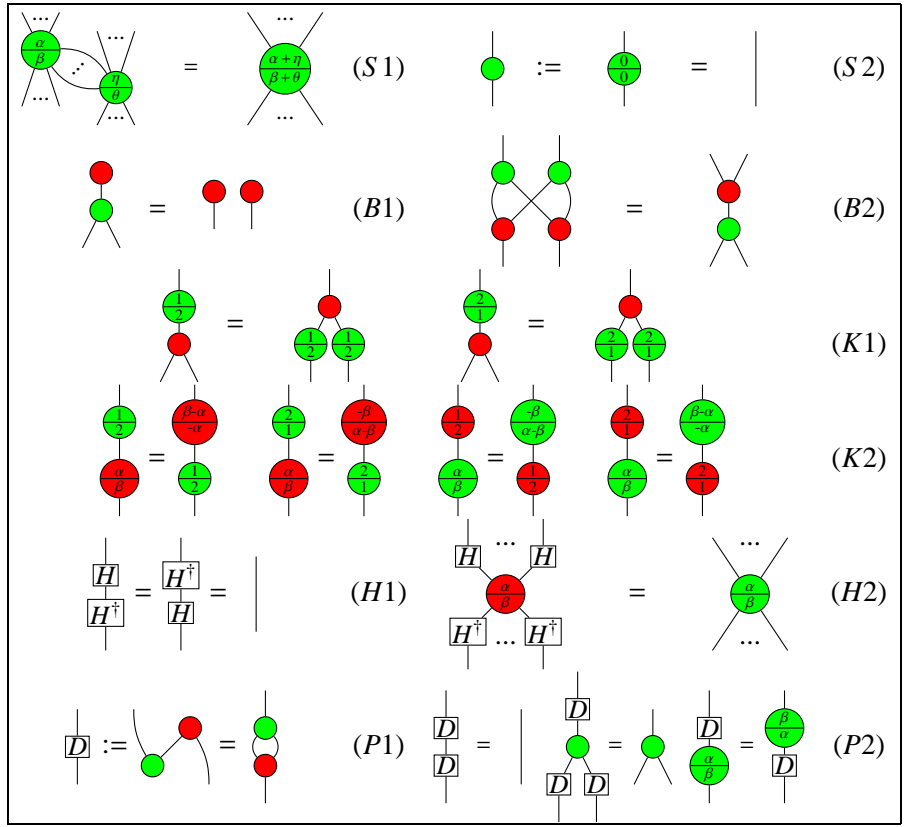

\begin{enumerate}
    \item Equations in Figure\ref{figure1}.
 \item  All equations hold under flip of graphs,  negation of angles, and exchange of $H$ and $H^\dagger$.\label{dagger}
  \item  All equations hold under flip of colours (except for rules $K2$ and $H2$).
\end{enumerate}

The equations below can be derived from the rules
of $\mathbf{RG}$ given above. They are very useful when demonstrating  some
more complex equalities in describing quantum protocols \cite{wangbian} and algorithms\cite{Gedik}.

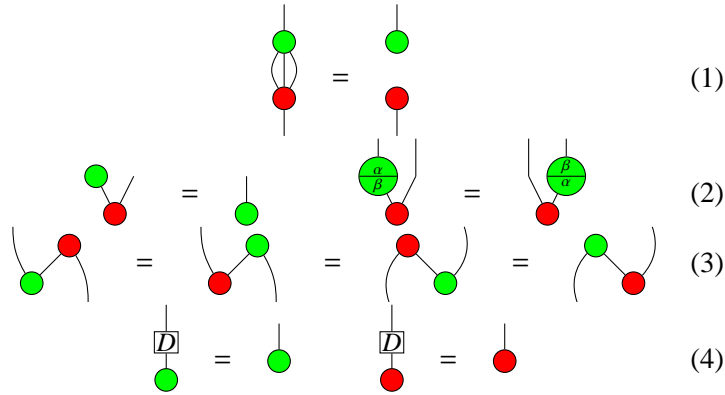
\begin{center}
\[
\quad \qquad\begin{tabular}{cc}\label{derivedRules}
\beginpgfgraphicnamed{RGrelations/deri1_3break}
\begin{tikzpicture}
	\begin{pgfonlayer}{nodelayer}
		\node [style=rn] (0) at (0.75, -0.25) {};
		\node [style=none] (1) at (-0.75, 1) {};
		\node [style=none] (2) at (0, -0) {$=$};
		\node [style=none] (3) at (0.75, -0.75) {};
		\node [style=gn] (4) at (-0.75, 0.5) {};
		\node [style=none] (5) at (0.75, 1) {};
		\node [style=none] (6) at (-0.75, -0.75) {};
		\node [style=gn] (7) at (0.75, 0.5) {};
		\node [style=rn] (8) at (-0.75, -0.25) {};
	\end{pgfonlayer}
	\begin{pgfonlayer}{edgelayer}
		\draw [style=none, in=120, out=-120, looseness=1.25] (4) to (8);
		\draw [style=none, bend right, looseness=1.25] (8) to (4);
		\draw [style=none] (4) to (1.center);
		\draw [style=none] (6.center) to (8);
		\draw [style=none] (3.center) to (0);
		\draw [style=none] (7) to (5.center);
		\draw [style=none] (4) to (8);
	\end{pgfonlayer}
\end{tikzpicture}}
\endpgfgraphicnamed&(1)\\
\beginpgfgraphicnamed{RGrelations/deri2_rotate}
\begin{tikzpicture}
	\begin{pgfonlayer}{nodelayer}
		\node [style=none] (0) at (-0.75, 0.25) {};
		\node [style=none] (1) at (-2.25, 0.25) {};
		\node [style=none] (2) at (-1.5, -0) {$=$};
		\node [style=gn] (3) at (-2.75, 0.25) {};
		\node [style=gn] (4) at (-0.75, -0.25) {};
		\node [style=none] (5) at (2.25, -0) {$=$};
		\node [style=gsn] (6) at (1, 0.25) {$\alpha$\nodepart{lower}$\beta$};
		\node [style=gsn] (7) at (3.5, 0.25) {$\beta$\nodepart{lower}$\alpha$};
		\node [style=none] (8) at (3, 0.25) {};
		\node [style=none] (9) at (1.5, 0.25) {};
		\node [style=none] (10) at (1, 0.75) {};
		\node [style=none] (11) at (3, 0.75) {};
		\node [style=none] (12) at (1.5, 0.75) {};
		\node [style=none] (13) at (3.5, 0.75) {};
		\node [style=rn] (14) at (3.25, -0.25) {};
		\node [style=rn] (15) at (-2.5, -0.25) {};
		\node [style=rn] (16) at (1.25, -0.25) {};
	\end{pgfonlayer}
	\begin{pgfonlayer}{edgelayer}
		\draw (0.center) to (4);
		\draw (10.center) to (6);
		\draw (12.center) to (9.center);
		\draw (11.center) to (8.center);
		\draw (13.center) to (7);
		\draw (3) to (15);
		\draw (15) to (1.center);
		\draw (6) to (16);
		\draw (16) to (9.center);
		\draw (8.center) to (14);
		\draw (14) to (7);
	\end{pgfonlayer}
\end{tikzpicture}}
\endpgfgraphicnamed&(2)\\
\beginpgfgraphicnamed{RGrelations/deri3_4duliser}
\begin{tikzpicture}
	\begin{pgfonlayer}{nodelayer}
		\node [style=none] (0) at (-4.25, 0.5) {};
		\node [style=none] (1) at (-3.25, -0.5) {};
		\node [style=gn] (2) at (-4, -0.25) {};
		\node [style=none] (3) at (-2.5, -0) {$=$};
		\node [style=rn] (4) at (-3.5, 0.25) {};
		\node [style=none] (5) at (-1.75, 0.5) {};
		\node [style=none] (6) at (-0.75, -0.5) {};
		\node [style=rn] (7) at (-1.5, -0.25) {};
		\node [style=none] (8) at (0, -0) {$=$};
		\node [style=gn] (9) at (-1, 0.25) {};
		\node [style=none] (10) at (1.75, 0.5) {};
		\node [style=none] (11) at (0.75, -0.5) {};
		\node [style=gn] (12) at (1.5, -0.25) {};
		\node [style=rn] (13) at (1, 0.25) {};
		\node [style=none] (14) at (4.25, 0.5) {};
		\node [style=gn] (15) at (3.5, 0.25) {};
		\node [style=none] (16) at (3.25, -0.5) {};
		\node [style=none] (17) at (2.5, -0) {$=$};
		\node [style=rn] (18) at (4, -0.25) {};
	\end{pgfonlayer}
	\begin{pgfonlayer}{edgelayer}
		\draw [bend left=15, looseness=1.00] (2) to (0.center);
		\draw (2) to (4);
		\draw [bend left=15, looseness=1.00] (4) to (1.center);
		\draw [bend left=15, looseness=1.00] (7) to (5.center);
		\draw (7) to (9);
		\draw [bend left=15, looseness=1.00] (9) to (6.center);
		\draw [bend right, looseness=1.00] (12) to (10.center);
		\draw (12) to (13);
		\draw [bend right, looseness=1.00] (13) to (11.center);
		\draw [bend right, looseness=1.00] (18) to (14.center);
		\draw (18) to (15);
		\draw [bend right, looseness=1.00] (15) to (16.center);
	\end{pgfonlayer}
\end{tikzpicture}}
\endpgfgraphicnamed&(3)\\
\beginpgfgraphicnamed{RGrelations/deri4_ed}
\begin{tikzpicture}
	\begin{pgfonlayer}{nodelayer}
		\node [style=none] (0) at (-2.25, 0.75) {};
		\node [style={H box}] (1) at (-2.25, 0.25) {$D$};
		\node [style=none] (2) at (-0.75, 0.5) {};
		\node [style=none] (3) at (-1.5, -0) {$=$};
		\node [style=gn] (4) at (-0.75, -0) {};
		\node [style=gn] (5) at (-2.25, -0.25) {};
		\node [style={H box}] (6) at (0.75, 0.25) {$D$};
		\node [style=rn] (7) at (2.25, -0) {};
		\node [style=rn] (8) at (0.75, -0.25) {};
		\node [style=none] (9) at (1.5, -0) {$=$};
		\node [style=none] (10) at (0.75, 0.75) {};
		\node [style=none] (11) at (2.25, 0.5) {};
	\end{pgfonlayer}
	\begin{pgfonlayer}{edgelayer}
		\draw (5) to (1);
		\draw (1) to (0.center);
		\draw (4) to (2.center);
		\draw (8) to (6);
		\draw (6) to (10.center);
		\draw (7) to (11.center);
	\end{pgfonlayer}
\end{tikzpicture}}
\endpgfgraphicnamed&(4)
\end{tabular}\]
\end{center}

It is worth noting that there are some remarkable differences between qutrit rules  and qubit rules. First, in qubit case we have%
\beginpgfgraphicnamed{RGrelations/duliserE1}
\begin{tikzpicture}
	\begin{pgfonlayer}{nodelayer}
		\node [style=none] (0) at (-0.25, -0.5) {};
		\node [style=none] (1) at (-1, 0.5) {};
		\node [style=rn] (2) at (-0.5, 0.25) {};
		\node [style=none] (3) at (0, -0) {$=$};
		\node [style=none] (4) at (0.5, 0.5) {};
		\node [style=gn] (5) at (-0.75, -0.25) {};
		\node [style=none] (6) at (0.5, -0.5) {};
	\end{pgfonlayer}
	\begin{pgfonlayer}{edgelayer}
		\draw [bend left=15, looseness=1.00] (5) to (1.center);
		\draw (5) to (2);
		\draw [bend left=15, looseness=1.00] (2) to (0.center);
		\draw (4.center) to (6.center);
	\end{pgfonlayer}
\end{tikzpicture}}
\endpgfgraphicnamed, while in qutrit case we have%
\beginpgfgraphicnamed{RGrelations/duliserE2}
\begin{tikzpicture}
	\begin{pgfonlayer}{nodelayer}
		\node [style=rn] (0) at (0.5, -0.25) {};
		\node [style=none] (1) at (-0.25, -0.5) {};
		\node [style=gn] (2) at (0.5, 0.25) {};
		\node [style=none] (3) at (0.5, -0.5) {};
		\node [style=gn] (4) at (-0.75, -0.25) {};
		\node [style=rn] (5) at (-0.5, 0.25) {};
		\node [style=none] (6) at (0.5, 0.5) {};
		\node [style=none] (7) at (0, -0) {$=$};
		\node [style=none] (8) at (-1, 0.5) {};
	\end{pgfonlayer}
	\begin{pgfonlayer}{edgelayer}
		\draw [bend left=15, looseness=1.00] (4) to (8.center);
		\draw (2) to (6.center);
		\draw (4) to (5);
		\draw [bend left=15, looseness=1.00] (5) to (1.center);
		\draw (0) to (3.center);
		\draw [bend right=45, looseness=1.00] (2) to (0);
		\draw [bend left=45, looseness=1.00] (2) to (0);
	\end{pgfonlayer}
\end{tikzpicture}}
\endpgfgraphicnamed. Second, the dualizer of the two observables Z and X is an even permutation, i.e., the identical permutation. And there is only one odd permutation  %
\beginpgfgraphicnamed{RGrelations/xpi_qubit}
\begin{tikzpicture}
	\begin{pgfonlayer}{nodelayer}
		\node [style=none] (0) at (0, -0.5) {};
		\node [style=rn] (1) at (0, -0) {$\pi$};
		\node [style=none] (2) at (0, 0.5) {};
	\end{pgfonlayer}
	\begin{pgfonlayer}{edgelayer}
		\draw (0.center) to (1);
		\draw (1) to (2.center);
	\end{pgfonlayer}
\end{tikzpicture}}
\endpgfgraphicnamed in qubit case such that

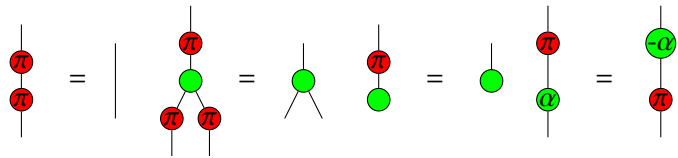
\ctikzfig{RGrelations/xpi_qubitEqu}

 While in qutrit case, the dualizer of Z and X is an odd permutation which satisfies rule $P2$. Third,  in qubit case the K2 rule still holds when flipping the colours, while it doesn't hold under flip of colours in qutrit case.

%

Now $\mathbf{RG}$ is a symmetric monoidal category, which can further
be made into a $\dagger-$ SMC by having $\dagger$ act on
the generators as follows:

\begin{center}
\[
\left(%
\beginpgfgraphicnamed{RGgenerator/RGg_ezd}
\begin{tikzpicture}
	\begin{pgfonlayer}{nodelayer}
		\node [style=none] (0) at (0, -0.5) {};
		\node [style=gn] (1) at (0, -0) {};
	\end{pgfonlayer}
	\begin{pgfonlayer}{edgelayer}
		\draw (0.center) to (1);
	\end{pgfonlayer}
\end{tikzpicture}}
\endpgfgraphicnamed\right)^\dagger=%
\beginpgfgraphicnamed{RGgenerator/RGg_ez}
\begin{tikzpicture}
	\begin{pgfonlayer}{nodelayer}
		\node [style=none] (0) at (0, 0.5) {};
		\node [style=gn] (1) at (0, -0) {};
	\end{pgfonlayer}
	\begin{pgfonlayer}{edgelayer}
		\draw (0.center) to (1);
	\end{pgfonlayer}
\end{tikzpicture}}
\endpgfgraphicnamed \qquad
\left(%
\beginpgfgraphicnamed{RGgenerator/RGg_ez}
\begin{tikzpicture}
	\begin{pgfonlayer}{nodelayer}
		\node [style=none] (0) at (0, 0.5) {};
		\node [style=gn] (1) at (0, -0) {};
	\end{pgfonlayer}
	\begin{pgfonlayer}{edgelayer}
		\draw (0.center) to (1);
	\end{pgfonlayer}
\end{tikzpicture}}
\endpgfgraphicnamed\right)^\dagger=%
\beginpgfgraphicnamed{RGgenerator/RGg_ezd}
\begin{tikzpicture}
	\begin{pgfonlayer}{nodelayer}
		\node [style=none] (0) at (0, -0.5) {};
		\node [style=gn] (1) at (0, -0) {};
	\end{pgfonlayer}
	\begin{pgfonlayer}{edgelayer}
		\draw (0.center) to (1);
	\end{pgfonlayer}
\end{tikzpicture}}
\endpgfgraphicnamed \qquad
\left(%
\beginpgfgraphicnamed{RGgenerator/RGg_dzd}
\begin{tikzpicture}
	\begin{pgfonlayer}{nodelayer}
		\node [style=gn] (0) at (0, -0) {};
		\node [style=none] (1) at (0, -0.5) {};
		\node [style=none] (2) at (0.25, 0.5) {};
		\node [style=none] (3) at (-0.25, 0.5) {};
	\end{pgfonlayer}
	\begin{pgfonlayer}{edgelayer}
		\draw (0) to (1.center);
		\draw (0) to (2.center);
		\draw (0) to (3.center);
	\end{pgfonlayer}
\end{tikzpicture}}
\endpgfgraphicnamed\right)^\dagger=%
\beginpgfgraphicnamed{RGgenerator/RGg_dz}
\begin{tikzpicture}
	\begin{pgfonlayer}{nodelayer}
		\node [style=none] (0) at (0, 0.5) {};
		\node [style=none] (1) at (0.25, -0.5) {};
		\node [style=none] (2) at (-0.25, -0.5) {};
		\node [style=gn] (3) at (0, -0) {};
	\end{pgfonlayer}
	\begin{pgfonlayer}{edgelayer}
		\draw (0.center) to (3);
		\draw (3) to (2.center);
		\draw (3) to (1.center);
	\end{pgfonlayer}
\end{tikzpicture}}
\endpgfgraphicnamed \qquad
\left(%
\beginpgfgraphicnamed{RGgenerator/RGg_dz}
\begin{tikzpicture}
	\begin{pgfonlayer}{nodelayer}
		\node [style=none] (0) at (0, 0.5) {};
		\node [style=none] (1) at (0.25, -0.5) {};
		\node [style=none] (2) at (-0.25, -0.5) {};
		\node [style=gn] (3) at (0, -0) {};
	\end{pgfonlayer}
	\begin{pgfonlayer}{edgelayer}
		\draw (0.center) to (3);
		\draw (3) to (2.center);
		\draw (3) to (1.center);
	\end{pgfonlayer}
\end{tikzpicture}}
\endpgfgraphicnamed\right)^\dagger=%
\beginpgfgraphicnamed{RGgenerator/RGg_dzd}
\begin{tikzpicture}
	\begin{pgfonlayer}{nodelayer}
		\node [style=gn] (0) at (0, -0) {};
		\node [style=none] (1) at (0, -0.5) {};
		\node [style=none] (2) at (0.25, 0.5) {};
		\node [style=none] (3) at (-0.25, 0.5) {};
	\end{pgfonlayer}
	\begin{pgfonlayer}{edgelayer}
		\draw (0) to (1.center);
		\draw (0) to (2.center);
		\draw (0) to (3.center);
	\end{pgfonlayer}
\end{tikzpicture}}
\endpgfgraphicnamed \qquad
\left(%
\beginpgfgraphicnamed{RGgenerator/RGg_zph_ab}
\begin{tikzpicture}
	\begin{pgfonlayer}{nodelayer}
		\node [style=gsn] (0) at (0, -0) {$\alpha$\nodepart{lower}$\beta$};
		\node [style=none] (1) at (0, 0.5) {};
		\node [style=none] (2) at (0, -0.5) {};
	\end{pgfonlayer}
	\begin{pgfonlayer}{edgelayer}
		\draw (1.center) to (0);
		\draw (2.center) to (0);
	\end{pgfonlayer}
\end{tikzpicture}}
\endpgfgraphicnamed\right)^\dagger=%
\beginpgfgraphicnamed{RGstruLevel1/RGs_zph_ab_minus}
\begin{tikzpicture}
	\begin{pgfonlayer}{nodelayer}
		\node [style=none] (0) at (0, -0.5) {};
		\node [style=gsn] (1) at (0, -0) {$\textnormal{-}\alpha$\nodepart{lower}$\textnormal{-}\beta$};
		\node [style=none] (2) at (0, 0.5) {};
	\end{pgfonlayer}
	\begin{pgfonlayer}{edgelayer}
		\draw (2.center) to (1);
		\draw (0.center) to (1);
	\end{pgfonlayer}
\end{tikzpicture}
}
\endpgfgraphicnamed\qquad
\left(%
\beginpgfgraphicnamed{RGgenerator/RGg_Hada}
\begin{tikzpicture}
	\begin{pgfonlayer}{nodelayer}
		\node [style={H box}] (0) at (0, -0) {$H$};
		\node [style=none] (1) at (0, 0.5) {};
		\node [style=none] (2) at (0, -0.5) {};
	\end{pgfonlayer}
	\begin{pgfonlayer}{edgelayer}
		\draw (1.center) to (0);
		\draw (2.center) to (0);
	\end{pgfonlayer}
\end{tikzpicture}}
\endpgfgraphicnamed\right)^\dagger=%
\beginpgfgraphicnamed{RGgenerator/RGg_Hadad}
\begin{tikzpicture}
	\begin{pgfonlayer}{nodelayer}
		\node [style={H box}] (0) at (0, -0) {$H^\dagger$};
		\node [style=none] (1) at (0, -0.5) {};
		\node [style=none] (2) at (0, 0.5) {};
	\end{pgfonlayer}
	\begin{pgfonlayer}{edgelayer}
		\draw (2.center) to (0);
		\draw (1.center) to (0);
	\end{pgfonlayer}
\end{tikzpicture}}
\endpgfgraphicnamed
\]

\[
\left(%
\beginpgfgraphicnamed{RGgenerator/RGg_exd}
\begin{tikzpicture}
	\begin{pgfonlayer}{nodelayer}
		\node [style=none] (0) at (0, -0.5) {};
		\node [style=rn] (1) at (0, -0) {};
	\end{pgfonlayer}
	\begin{pgfonlayer}{edgelayer}
		\draw (0.center) to (1);
	\end{pgfonlayer}
\end{tikzpicture}}
\endpgfgraphicnamed\right)^\dagger=%
\beginpgfgraphicnamed{RGgenerator/RGg_ex}
\begin{tikzpicture}
	\begin{pgfonlayer}{nodelayer}
		\node [style=none] (0) at (0, 0.5) {};
		\node [style=rn] (1) at (0, -0) {};
	\end{pgfonlayer}
	\begin{pgfonlayer}{edgelayer}
		\draw (0.center) to (1);
	\end{pgfonlayer}
\end{tikzpicture}}
\endpgfgraphicnamed \qquad
\left(%
\beginpgfgraphicnamed{RGgenerator/RGg_ex}
\begin{tikzpicture}
	\begin{pgfonlayer}{nodelayer}
		\node [style=none] (0) at (0, 0.5) {};
		\node [style=rn] (1) at (0, -0) {};
	\end{pgfonlayer}
	\begin{pgfonlayer}{edgelayer}
		\draw (0.center) to (1);
	\end{pgfonlayer}
\end{tikzpicture}}
\endpgfgraphicnamed\right)^\dagger=%
\beginpgfgraphicnamed{RGgenerator/RGg_exd}
\begin{tikzpicture}
	\begin{pgfonlayer}{nodelayer}
		\node [style=none] (0) at (0, -0.5) {};
		\node [style=rn] (1) at (0, -0) {};
	\end{pgfonlayer}
	\begin{pgfonlayer}{edgelayer}
		\draw (0.center) to (1);
	\end{pgfonlayer}
\end{tikzpicture}}
\endpgfgraphicnamed \qquad
\left(%
\beginpgfgraphicnamed{RGgenerator/RGg_dxd}
\begin{tikzpicture}
	\begin{pgfonlayer}{nodelayer}
		\node [style=rn] (0) at (0, -0) {};
		\node [style=none] (1) at (0, -0.5) {};
		\node [style=none] (2) at (0.25, 0.5) {};
		\node [style=none] (3) at (-0.25, 0.5) {};
	\end{pgfonlayer}
	\begin{pgfonlayer}{edgelayer}
		\draw (0) to (1.center);
		\draw (0) to (2.center);
		\draw (0) to (3.center);
	\end{pgfonlayer}
\end{tikzpicture}}
\endpgfgraphicnamed\right)^\dagger=%
\beginpgfgraphicnamed{RGgenerator/RGg_dx}
\begin{tikzpicture}
	\begin{pgfonlayer}{nodelayer}
		\node [style=none] (0) at (0, 0.5) {};
		\node [style=none] (1) at (0.25, -0.5) {};
		\node [style=none] (2) at (-0.25, -0.5) {};
		\node [style=rn] (3) at (0, -0) {};
	\end{pgfonlayer}
	\begin{pgfonlayer}{edgelayer}
		\draw (0.center) to (3);
		\draw (3) to (2.center);
		\draw (3) to (1.center);
	\end{pgfonlayer}
\end{tikzpicture}}
\endpgfgraphicnamed \qquad
\left(%
\beginpgfgraphicnamed{RGgenerator/RGg_dx}
\begin{tikzpicture}
	\begin{pgfonlayer}{nodelayer}
		\node [style=none] (0) at (0, 0.5) {};
		\node [style=none] (1) at (0.25, -0.5) {};
		\node [style=none] (2) at (-0.25, -0.5) {};
		\node [style=rn] (3) at (0, -0) {};
	\end{pgfonlayer}
	\begin{pgfonlayer}{edgelayer}
		\draw (0.center) to (3);
		\draw (3) to (2.center);
		\draw (3) to (1.center);
	\end{pgfonlayer}
\end{tikzpicture}}
\endpgfgraphicnamed\right)^\dagger=%
\beginpgfgraphicnamed{RGgenerator/RGg_dxd}
\begin{tikzpicture}
	\begin{pgfonlayer}{nodelayer}
		\node [style=rn] (0) at (0, -0) {};
		\node [style=none] (1) at (0, -0.5) {};
		\node [style=none] (2) at (0.25, 0.5) {};
		\node [style=none] (3) at (-0.25, 0.5) {};
	\end{pgfonlayer}
	\begin{pgfonlayer}{edgelayer}
		\draw (0) to (1.center);
		\draw (0) to (2.center);
		\draw (0) to (3.center);
	\end{pgfonlayer}
\end{tikzpicture}}
\endpgfgraphicnamed \qquad
\left(%
\beginpgfgraphicnamed{RGgenerator/RGg_xph_ab}
\begin{tikzpicture}
	\begin{pgfonlayer}{nodelayer}
		\node [style=rsn] (0) at (0, -0) {$\alpha$\nodepart{lower}$\beta$};
		\node [style=none] (1) at (0, 0.5) {};
		\node [style=none] (2) at (0, -0.5) {};
	\end{pgfonlayer}
	\begin{pgfonlayer}{edgelayer}
		\draw (1.center) to (0);
		\draw (2.center) to (0);
	\end{pgfonlayer}
\end{tikzpicture}}
\endpgfgraphicnamed\right)^\dagger=%
\beginpgfgraphicnamed{RGstruLevel1/RGs_xph_ab_minus}
\begin{tikzpicture}
	\begin{pgfonlayer}{nodelayer}
		\node [style=rsn] (0) at (0, -0) {$\textnormal{-}\alpha$\nodepart{lower}$\textnormal{-}\beta$};
		\node [style=none] (1) at (0, 0.5) {};
		\node [style=none] (2) at (0, -0.5) {};
	\end{pgfonlayer}
	\begin{pgfonlayer}{edgelayer}
		\draw (1.center) to (0);
		\draw (2.center) to (0);
	\end{pgfonlayer}
\end{tikzpicture}
}
\endpgfgraphicnamed \qquad
\left(%
\beginpgfgraphicnamed{RGgenerator/RGg_Hadad}
\begin{tikzpicture}
	\begin{pgfonlayer}{nodelayer}
		\node [style={H box}] (0) at (0, -0) {$H^\dagger$};
		\node [style=none] (1) at (0, -0.5) {};
		\node [style=none] (2) at (0, 0.5) {};
	\end{pgfonlayer}
	\begin{pgfonlayer}{edgelayer}
		\draw (2.center) to (0);
		\draw (1.center) to (0);
	\end{pgfonlayer}
\end{tikzpicture}}
\endpgfgraphicnamed\right)^\dagger=%
\beginpgfgraphicnamed{RGgenerator/RGg_Hada}
\begin{tikzpicture}
	\begin{pgfonlayer}{nodelayer}
		\node [style={H box}] (0) at (0, -0) {$H$};
		\node [style=none] (1) at (0, 0.5) {};
		\node [style=none] (2) at (0, -0.5) {};
	\end{pgfonlayer}
	\begin{pgfonlayer}{edgelayer}
		\draw (1.center) to (0);
		\draw (2.center) to (0);
	\end{pgfonlayer}
\end{tikzpicture}}
\endpgfgraphicnamed
\]
\end{center}

where functoriality of $(\cdot)^\dagger$ is guaranteed by Rule \ref{dagger}.

\subsection{RG interpretation}
 Here, we give an
interpretation for these graphs by describing a
monoidal functor $[\cdot]_{RG}:\mathbf{RG}\rightarrow \mathbf{FdHilb_Q}$, mapping the morphisms as follows (expressed
in Dirac notation):
\begin{center}
\[
\left[%
\beginpgfgraphicnamed{RGgenerator/RGg_ezd}
\begin{tikzpicture}
	\begin{pgfonlayer}{nodelayer}
		\node [style=none] (0) at (0, -0.5) {};
		\node [style=gn] (1) at (0, -0) {};
	\end{pgfonlayer}
	\begin{pgfonlayer}{edgelayer}
		\draw (0.center) to (1);
	\end{pgfonlayer}
\end{tikzpicture}}
\endpgfgraphicnamed\right]_{RG}=\ket{+} \qquad
\left[%
\beginpgfgraphicnamed{RGgenerator/RGg_ez}
\begin{tikzpicture}
	\begin{pgfonlayer}{nodelayer}
		\node [style=none] (0) at (0, 0.5) {};
		\node [style=gn] (1) at (0, -0) {};
	\end{pgfonlayer}
	\begin{pgfonlayer}{edgelayer}
		\draw (0.center) to (1);
	\end{pgfonlayer}
\end{tikzpicture}}
\endpgfgraphicnamed\right]_{RG}=\bra{+} \qquad
\left[%
\beginpgfgraphicnamed{RGgenerator/RGg_dz}
\begin{tikzpicture}
	\begin{pgfonlayer}{nodelayer}
		\node [style=none] (0) at (0, 0.5) {};
		\node [style=none] (1) at (0.25, -0.5) {};
		\node [style=none] (2) at (-0.25, -0.5) {};
		\node [style=gn] (3) at (0, -0) {};
	\end{pgfonlayer}
	\begin{pgfonlayer}{edgelayer}
		\draw (0.center) to (3);
		\draw (3) to (2.center);
		\draw (3) to (1.center);
	\end{pgfonlayer}
\end{tikzpicture}}
\endpgfgraphicnamed\right]_{RG}=\ket{00}\bra{0}+\ket{11}\bra{1}+\ket{22}\bra{2}
\]
\[
\left[%
\beginpgfgraphicnamed{RGgenerator/RGg_dzd}
\begin{tikzpicture}
	\begin{pgfonlayer}{nodelayer}
		\node [style=gn] (0) at (0, -0) {};
		\node [style=none] (1) at (0, -0.5) {};
		\node [style=none] (2) at (0.25, 0.5) {};
		\node [style=none] (3) at (-0.25, 0.5) {};
	\end{pgfonlayer}
	\begin{pgfonlayer}{edgelayer}
		\draw (0) to (1.center);
		\draw (0) to (2.center);
		\draw (0) to (3.center);
	\end{pgfonlayer}
\end{tikzpicture}}
\endpgfgraphicnamed\right]_{RG}=\ket{0}\bra{00}+\ket{1}\bra{11}+\ket{2}\bra{22} \qquad
\left[%
\beginpgfgraphicnamed{RGgenerator/RGg_zph_ab}
\begin{tikzpicture}
	\begin{pgfonlayer}{nodelayer}
		\node [style=gsn] (0) at (0, -0) {$\alpha$\nodepart{lower}$\beta$};
		\node [style=none] (1) at (0, 0.5) {};
		\node [style=none] (2) at (0, -0.5) {};
	\end{pgfonlayer}
	\begin{pgfonlayer}{edgelayer}
		\draw (1.center) to (0);
		\draw (2.center) to (0);
	\end{pgfonlayer}
\end{tikzpicture}}
\endpgfgraphicnamed\right]_{RG}=\ket{0}\bra{0}+ e^{i \alpha}\ket{1}\bra{1}
+ e^{i \beta}\ket{2}\bra{2} \qquad
\]

\[
\left[%
\beginpgfgraphicnamed{RGgenerator/RGg_exd}
\begin{tikzpicture}
	\begin{pgfonlayer}{nodelayer}
		\node [style=none] (0) at (0, -0.5) {};
		\node [style=rn] (1) at (0, -0) {};
	\end{pgfonlayer}
	\begin{pgfonlayer}{edgelayer}
		\draw (0.center) to (1);
	\end{pgfonlayer}
\end{tikzpicture}}
\endpgfgraphicnamed\right]_{RG}=\ket{0} \qquad
\left[%
\beginpgfgraphicnamed{RGgenerator/RGg_ex}
\begin{tikzpicture}
	\begin{pgfonlayer}{nodelayer}
		\node [style=none] (0) at (0, 0.5) {};
		\node [style=rn] (1) at (0, -0) {};
	\end{pgfonlayer}
	\begin{pgfonlayer}{edgelayer}
		\draw (0.center) to (1);
	\end{pgfonlayer}
\end{tikzpicture}}
\endpgfgraphicnamed\right]_{RG}=\bra{0} \qquad
\left[%
\beginpgfgraphicnamed{RGgenerator/RGg_dx}
\begin{tikzpicture}
	\begin{pgfonlayer}{nodelayer}
		\node [style=none] (0) at (0, 0.5) {};
		\node [style=none] (1) at (0.25, -0.5) {};
		\node [style=none] (2) at (-0.25, -0.5) {};
		\node [style=rn] (3) at (0, -0) {};
	\end{pgfonlayer}
	\begin{pgfonlayer}{edgelayer}
		\draw (0.center) to (3);
		\draw (3) to (2.center);
		\draw (3) to (1.center);
	\end{pgfonlayer}
\end{tikzpicture}}
\endpgfgraphicnamed\right]_{RG}=\ket{++}\bra{ +}+\ket{\omega\omega}\bra{\omega}+\ket{\bar{\omega}\bar{\omega}}\bra{\bar{\omega}}
\]
\[
\left[%
\beginpgfgraphicnamed{RGgenerator/RGg_dxd}
\begin{tikzpicture}
	\begin{pgfonlayer}{nodelayer}
		\node [style=rn] (0) at (0, -0) {};
		\node [style=none] (1) at (0, -0.5) {};
		\node [style=none] (2) at (0.25, 0.5) {};
		\node [style=none] (3) at (-0.25, 0.5) {};
	\end{pgfonlayer}
	\begin{pgfonlayer}{edgelayer}
		\draw (0) to (1.center);
		\draw (0) to (2.center);
		\draw (0) to (3.center);
	\end{pgfonlayer}
\end{tikzpicture}}
\endpgfgraphicnamed\right]_{RG}=\ket{+}\bra{++}+\ket{\omega}\bra{\omega\omega}
+\ket{\bar{\omega}}\bra{\bar{\omega}\bar{\omega}} \qquad
\left[%
\beginpgfgraphicnamed{RGgenerator/RGg_xph_ab}
\begin{tikzpicture}
	\begin{pgfonlayer}{nodelayer}
		\node [style=rsn] (0) at (0, -0) {$\alpha$\nodepart{lower}$\beta$};
		\node [style=none] (1) at (0, 0.5) {};
		\node [style=none] (2) at (0, -0.5) {};
	\end{pgfonlayer}
	\begin{pgfonlayer}{edgelayer}
		\draw (1.center) to (0);
		\draw (2.center) to (0);
	\end{pgfonlayer}
\end{tikzpicture}}
\endpgfgraphicnamed\right]_{RG}=\ket{+}\bra{+}+ e^{i \alpha}\ket{\omega}\bra{\omega}
+ e^{i \beta}\ket{\bar{\omega}}\bra{\bar{\omega}}
\]

\[
\left[%
\beginpgfgraphicnamed{RGgenerator/RGg_Hada}
\begin{tikzpicture}
	\begin{pgfonlayer}{nodelayer}
		\node [style={H box}] (0) at (0, -0) {$H$};
		\node [style=none] (1) at (0, 0.5) {};
		\node [style=none] (2) at (0, -0.5) {};
	\end{pgfonlayer}
	\begin{pgfonlayer}{edgelayer}
		\draw (1.center) to (0);
		\draw (2.center) to (0);
	\end{pgfonlayer}
\end{tikzpicture}}
\endpgfgraphicnamed\right]_{RG}=\ket{+}\bra{0}+ \ket{\omega}\bra{1}+\ket{\bar{\omega}}\bra{2} \qquad
\left[%
\beginpgfgraphicnamed{RGgenerator/RGg_Hadad}
\begin{tikzpicture}
	\begin{pgfonlayer}{nodelayer}
		\node [style={H box}] (0) at (0, -0) {$H^\dagger$};
		\node [style=none] (1) at (0, -0.5) {};
		\node [style=none] (2) at (0, 0.5) {};
	\end{pgfonlayer}
	\begin{pgfonlayer}{edgelayer}
		\draw (2.center) to (0);
		\draw (1.center) to (0);
	\end{pgfonlayer}
\end{tikzpicture}}
\endpgfgraphicnamed\right]_{RG}=\ket{0}\bra{+}+ \ket{1}\bra{\omega}+\ket{2}\bra{\bar{\omega}}
\]
\end{center}

where $\omega=e^{\frac{2}{3}\pi i},\bar{\omega}=e^{\frac{4}{3}\pi i}$, and
 \begin{center}
 $
   \left\{\begin{array}{rcl}
      \ket{+} & = & \ket{0}+\ket{1}+\ket{2}\\
      \ket{\omega} & = &  \ket{0}+\omega \ket{1}+\bar{\omega}\ket{2}\\
      \ket{\bar{\omega}} & = & \ket{0}+\bar{\omega}\ket{1}+\omega\ket{2}
    \end{array}\right.
   $
 \end{center}

\begin{proposition}\label{inter}
 $[\cdot]_{RG}$ is  a symmetric monoidal $\dag-$functor.
\end{proposition}
\begin{proof}
This involves checking for each rule $f = g$
in $\mathbf{RG}$ that $[f]_{RG}=[g]_{RG}$, that$[\cdot]_{RG}$ respects
the symmetric monoidal structure on the generators,
and for each generator $f$, we have $[f]^{\dag}_{RG}=[f^{\dag}]_{RG}$.
\end{proof}

\section{Decomposition of the Hadamard Gate}\label{sec5}
It can be directly checked that in $\mathbf{FdHilb_Q}$, we have
\begin{equation*}
[H]_{RG}= [P_X(\frac{4\pi}{3},\frac{4\pi}{3})]_{RG}\circ[P_Z(\frac{4\pi}{3},\frac{4\pi}{3})]_{RG}\circ[P_X(\frac{4\pi}{3},\frac{4\pi}{3})]_{RG}
\end{equation*}
We call the following graph an Euler decomposition of
the Hadamard gate:
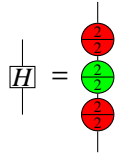
\ctikzfig{HadaDecomSingle}

\begin{proposition}
The Euler decomposition is not unique:
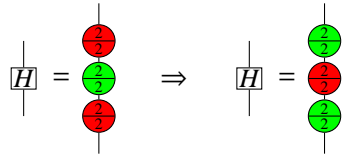
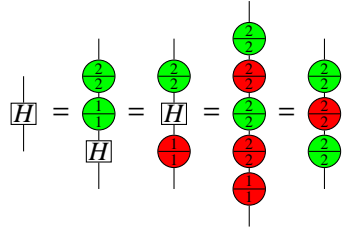
\ctikzfig{HadaDecom}
\end{proposition}

\begin{proof}
\ctikzfig{HadaDecomProof}

\end{proof}
In the qubit case, Duncan and Perdrix \cite{RossPerdrix} proved that the Euler decomposition is not derivable from  ZX calculus. Similarly, we have
\begin{proposition}
The Euler decomposition is not derivable from $\mathbf{RG}$.
\end{proposition}

\begin{proof}
We define an alternative interpretation functor $[\cdot]_0: \mathbf{RG} \rightarrow \mathbf{FdHilb_Q}$ exactly as $[\cdot]_{RG}$ with the following change:

\begin{equation*}
[P_X(\alpha,\beta)]_{0}=[P_X(0,0)]_{RG} \qquad [P_Z(\alpha,\beta)]_{0}=[P_Z(0,0)]_{RG}
\end{equation*}

This functor preserves all the rules introduced in Figure \ref{figure1}, so its image is indeed a valid model of the theory.
However we have the following inequality
\begin{equation*}
[H]_{0}\neq [P_X(\frac{4\pi}{3},\frac{4\pi}{3})]_{0}\circ[P_Z(\frac{4\pi}{3},\frac{4\pi}{3})]_{0}\circ[P_X(\frac{4\pi}{3},\frac{4\pi}{3})]_{0}
\end{equation*}
 hence the Euler decomposition is not derivable from $\mathbf{RG}$.

\end{proof}
\section{Quantum Algorithm with a Single Qutrit}
Recently, Gedik\cite{Gedik} introduces
a simple algorithm using only a single qutrit to determine the parity
of permutations of a set of three objects. As in the case of Deutsch's algorithm, a speed-up relative to corresponding classical algorithms is obtained.

Consider the six permutations of the set $\{0,1,2\}$. Each permutation can be treated as a function $f(x)$ defined on the set $x\in \{0, 1, 2\}$. Then the task is to determine its parity. The problem could be solved by evaluating $f(x)$ for two different values of x.

The function $f$ has a domain and range of three values. These three values correspond to the three states
of a qutrit $\ket{m}$ where $m = 0,1,2$. The unitary $U_f$ corresponding to
the function $f$ is a simple transposition of orthonormal states $\ket{m}$. Applying $U_f$ to the eigenstate $\ket{\omega}$ of the $X$ observable we obtain

 \begin{center}
 $
   \left\{\begin{array}{rcll}
      U_f \ket{\omega} & = &\ket{\omega}\textnormal{(up to a phase)} & \textnormal{if $f$ is an even permutation;}\\
      U_f \ket{\omega} & = &\ket{\bar{\omega}}\textnormal{(up to a phase)} & \textnormal{if $f$ is an odd permutation.}
    \end{array}\right.
   $
 \end{center}

Thus, a single evaluation of the function is enough to determine its parity.

The above algorithm can be depicted by the dichromatic calculus as follows:

\[
\begin{tabular}{|c|c|@{}c@{}|@{}c@{}|c|c|c|}
\hline
$f$ &(0)&(1 2)(0 1)&(1 2)(0 2)&(1 2)&(0 1)&(0 2)\\
\hline
$U_f$ & %
\beginpgfgraphicnamed{Algorithm/p00}
\begin{tikzpicture}
	\begin{pgfonlayer}{nodelayer}
		\node [style=none] (0) at (0, -0.5) {};
		\node [style=none] (1) at (0, 0.5) {};
		\node [style=rsn] (2) at (0, -0) {\tiny $0$\nodepart{lower}\tiny $0$};
	\end{pgfonlayer}
	\begin{pgfonlayer}{edgelayer}
		\draw (2) to (0.center);
		\draw (2) to (1.center);
	\end{pgfonlayer}
\end{tikzpicture}}
\endpgfgraphicnamed & %
\beginpgfgraphicnamed{Algorithm/p12}
\begin{tikzpicture}
	\begin{pgfonlayer}{nodelayer}
		\node [style=none] (0) at (0, 0.5) {};
		\node [style=rsn] (1) at (0, -0) {\tiny $1$\nodepart{lower}\tiny $2$};
		\node [style=none] (2) at (0, -0.5) {};
	\end{pgfonlayer}
	\begin{pgfonlayer}{edgelayer}
		\draw (1) to (2.center);
		\draw (1) to (0.center);
	\end{pgfonlayer}
\end{tikzpicture}}
\endpgfgraphicnamed & %
\beginpgfgraphicnamed{Algorithm/p21}
\begin{tikzpicture}
	\begin{pgfonlayer}{nodelayer}
		\node [style=none] (0) at (0, 0.5) {};
		\node [style=rsn] (1) at (0, -0) {\tiny $2$\nodepart{lower}\tiny $1$};
		\node [style=none] (2) at (0, -0.5) {};
	\end{pgfonlayer}
	\begin{pgfonlayer}{edgelayer}
		\draw (1) to (2.center);
		\draw (1) to (0.center);
	\end{pgfonlayer}
\end{tikzpicture}}
\endpgfgraphicnamed & %
\beginpgfgraphicnamed{Algorithm/pd}
\begin{tikzpicture}
	\begin{pgfonlayer}{nodelayer}
		\node [style={H box}] (0) at (0, -0.25) {$D$};
		\node [style=none] (1) at (0, -0.75) {};
		\node [style=none] (2) at (0, 0.25) {};
	\end{pgfonlayer}
	\begin{pgfonlayer}{edgelayer}
		\draw (1.center) to (0);
		\draw (0) to (2.center);
	\end{pgfonlayer}
\end{tikzpicture}}
\endpgfgraphicnamed & %
\beginpgfgraphicnamed{Algorithm/p12d}
\begin{tikzpicture}
	\begin{pgfonlayer}{nodelayer}
		\node [style=none] (0) at (0, -1) {};
		\node [style={H box}] (1) at (0, -0.5) {$D$};
		\node [style=rsn] (2) at (0, -0) {\tiny $1$\nodepart{lower}\tiny $2$};
		\node [style=none] (3) at (0, 0.5) {};
	\end{pgfonlayer}
	\begin{pgfonlayer}{edgelayer}
		\draw (2) to (1);
		\draw (1) to (0.center);
		\draw (2) to (3.center);
	\end{pgfonlayer}
\end{tikzpicture}}
\endpgfgraphicnamed & %
\beginpgfgraphicnamed{Algorithm/p21d}
\begin{tikzpicture}
	\begin{pgfonlayer}{nodelayer}
		\node [style=none] (0) at (0, -1) {};
		\node [style=rsn] (1) at (0, -0) {\tiny $2$\nodepart{lower}\tiny $1$};
		\node [style={H box}] (2) at (0, -0.5) {$D$};
		\node [style=none] (3) at (0, 0.5) {};
	\end{pgfonlayer}
	\begin{pgfonlayer}{edgelayer}
		\draw (1) to (2);
		\draw (2) to (0.center);
		\draw (1) to (3.center);
	\end{pgfonlayer}
\end{tikzpicture}}
\endpgfgraphicnamed\\
\hline
\textnormal{ $U_f \ket{w}$}&\multicolumn{3}{|c}{%
\beginpgfgraphicnamed{Algorithm/even}
\begin{tikzpicture}
	\begin{pgfonlayer}{nodelayer}
		\node [style=rsn] (0) at (-1, -0) {\tiny $0$\nodepart{lower}\tiny $0$};
		\node [style=gsn] (1) at (-1, 0.5) {\tiny $1$\nodepart{lower}\tiny $2$};
		\node [style=none] (2) at (-1, -0.5) {};
		\node [style=rsn] (3) at (1, -0) {\tiny $2$\nodepart{lower}\tiny $1$};
		\node [style=none] (4) at (1, -0.5) {};
		\node [style=rsn] (5) at (0, -0) {\tiny $1$\nodepart{lower}\tiny $2$};
		\node [style=none] (6) at (0, -0.5) {};
		\node [style=none] (7) at (-0.5, -0) {$=$};
		\node [style=none] (8) at (0.5, -0) {$=$};
		\node [style=none] (9) at (1.5, -0) {$=$};
		\node [style=gsn] (10) at (1, 0.5) {\tiny $1$\nodepart{lower}\tiny $2$};
		\node [style=gsn] (11) at (0, 0.5) {\tiny $1$\nodepart{lower}\tiny $2$};
		\node [style=gsn] (12) at (2, 0.25) {\tiny $1$\nodepart{lower}\tiny $2$};
		\node [style=none] (13) at (2, -0.25) {};
	\end{pgfonlayer}
	\begin{pgfonlayer}{edgelayer}
		\draw (0) to (2.center);
		\draw (0) to (1);
		\draw (3) to (4.center);
		\draw (5) to (6.center);
		\draw (3) to (10);
		\draw (5) to (11);
		\draw (12) to (13.center);
	\end{pgfonlayer}
\end{tikzpicture}}
\endpgfgraphicnamed}&\multicolumn{3}{|c|}{%
\beginpgfgraphicnamed{Algorithm/odd}
\begin{tikzpicture}
	\begin{pgfonlayer}{nodelayer}
		\node [style=gsn] (0) at (-0.5, 0.5) {\tiny $1$\nodepart{lower}\tiny $2$};
		\node [style=none] (1) at (1, -0) {$=$};
		\node [style=gsn] (2) at (1.5, 0.25) {\tiny $2$\nodepart{lower}\tiny $1$};
		\node [style=gsn] (3) at (0.5, 0.5) {\tiny $1$\nodepart{lower}\tiny $2$};
		\node [style=gsn] (4) at (-1.5, 0.25) {\tiny $1$\nodepart{lower}\tiny $2$};
		\node [style=none] (5) at (1.5, -0.25) {};
		\node [style=rsn] (6) at (0.5, -0) {\tiny $2$\nodepart{lower}\tiny $1$};
		\node [style=none] (7) at (0, -0) {$=$};
		\node [style=rsn] (8) at (-0.5, -0) {\tiny $1$\nodepart{lower}\tiny $2$};
		\node [style=none] (9) at (-1, -0) {$=$};
		\node [style={H box}] (10) at (-1.5, -0.25) {$D$};
		\node [style={H box}] (11) at (0.5, -0.5) {$D$};
		\node [style={H box}] (12) at (-0.5, -0.5) {$D$};
		\node [style=none] (13) at (-0.5, -1) {};
		\node [style=none] (14) at (0.5, -1) {};
		\node [style=none] (15) at (-1.5, -0.75) {};
	\end{pgfonlayer}
	\begin{pgfonlayer}{edgelayer}
		\draw (6) to (3);
		\draw (8) to (0);
		\draw (2) to (5.center);
		\draw (15.center) to (10);
		\draw (14.center) to (11);
		\draw (11) to (6);
		\draw (12) to (8);
		\draw (13.center) to (12);
		\draw (10) to (4);
	\end{pgfonlayer}
\end{tikzpicture}}
\endpgfgraphicnamed}\\
\hline
\textnormal{Parity}&\multicolumn{3}{|c}{Even}&\multicolumn{3}{|c|}{Odd}\\
\hline
\end{tabular}
\]

\section{The Qudit ZX Calculus Is Universal}

 It is  important to prove that the qudit ZX calculus is universal for quantum mechanics for any $d$.  Since it is not easy to decompose an arbitrary $d\times d$ unitary matrix into Z and X phase gates (i.e., $\Lambda_Z$ and $\Lambda_X$ gates) when $d>2$, the proof of university is far from trivial. Due to Brylinski \cite{Bry}, to prove the universality of qudit ZX calculus for quantum mechanics, it suffices to prove that the qudit ZX calculus contains all single qudit unitary transformations. Such a proof given in \cite{Ranchin} is based on the fact \cite{Muth} that the d-dimensional phase gates $Z_d, X_d$  are sufficient to simulate all single qudit unitary transforms, where \[Z_{d}(b_0,b_1...,b_{d-1}):b_0\ket{0}+b_1\ket{1}+...+b_{d-1}\ket{d-1}\mapsto \ket{d-1}\] (the $d$ complex coefficients, $b_0,b_1...,b_{d-1}$ are normalized to unity)
\begin{displaymath}
X_d(\phi): \left\{ \begin{array}{lll}
\ket{d-1} &\mapsto &e^{i \phi}\ket{d-1}\\
\ket{p} &\mapsto &\ket{p} \textnormal{ for } p\neq d-1
\end{array} \right.
\end{displaymath}

It was checked in \cite{Ranchin} that each  $X_d$ can be encoded to a phase gate $\Lambda_Z$  of the qudit ZX calculus, where
\begin{displaymath}
\Lambda_Z(\alpha_1,\alpha_2,
...,\alpha_{d-1}) :=\left(
\begin{array}{cccc}
1&                &        & \\
 & e^{i \alpha_1} &        & \\
 &                & \ddots & \\
 &                &        & e^{i \alpha_{d-1}}
\end{array}
\right)
\end{displaymath}

Meanwhile, some $Z_d$ phase gates were shown to be realized by  $\Lambda_X$ phase gate in the qudit ZX calculus, where

\begin{displaymath}
\Lambda_X(\alpha_1,\alpha_2,
...,\alpha_{d-1}) :=\frac{1}{d}\left(
\begin{array}{cccccc}
c_0     & c_{d-1} & c_{d-2} & ... & c_2     & c_1 \\
c_1     & c_0     & c_{d-1} & ... & c_3     & c_2 \\
c_2     & c_1     & c_0     & ... & c_4     & c_3 \\
...     & ...     & ...     & ... & ...     & ... \\
c_{d-1} & c_{d-2} & c_{d-3} & ... & c_1     & c_0
\end{array}
\right)
\end{displaymath}

$c_k=1+\sum^{d-1}_{l=1}\eta^{r_k(l)} e^{i\alpha_l}$, $r_k$
permutes the entries $1$ (there is one $r_k$ for each k).


 However, not every $Z_d$ phase gate can be represented by
$\Lambda_X(\alpha_1,\alpha_2,...,\alpha_{d-1})$. In fact, to realize any $Z_{d}(b_0,b_1...,b_{d-1})$ in this way,
we need  to find  $\alpha_1,\alpha_2,...,\alpha_{d-1}$ such that

\begin{eqnarray}
  c_{d-1} b_0+c_{d-2} b_1+...c_1 b_{d-2}+c_{0} b_{d-1} &=& d\phantom{,\quad \forall k\neq d-1}\label{equ1} \\
  c_k b_0+c_{k-1} b_1+...c_0 b_k+c_{d-1} b_{k+1}+...c_{k+1} b_{d-1} &=& 0,\quad \forall k\neq d-1 \label{equ2}
\end{eqnarray}


Since $\sum^{d-1}_{k=0} c_k=d$, summing up all the equations in (\ref{equ1}) and (\ref{equ2}),
we have $\sum^{d-1}_{k=0} b_k=1$. Clearly, not every unit complex vector $(b_0,b_1...,b_{d-1})$ satisfies $\sum^{d-1}_{k=0} b_k=1$ or $\sum^{d-1}_{k=0} b_k=e^{i \alpha}$ up to
a global phase.

For example, $(b_0,b_1...,b_{d-1})=(0, \frac{1}{\sqrt{2}},\frac{1}{\sqrt{2}},0,...,0), d>2$, is such a counterexample.

\vspace{5mm}
The above argument means that we need to find
another proof that the qudit ZX calculus contains all single qudit unitary transformations. Next we solve this problem using the theory of Lie algebra.

Let \begin{displaymath}
H =\left\{\left.\left(
\begin{array}{ccc}
e^{i \alpha_0} & \phantom{e^{i \alpha_0}} & \phantom{e^{i \alpha_0}} \\
\phantom{e^{i \alpha_0}} & \ddots & \phantom{e^{i \alpha_0}}\\
\phantom{e^{i \alpha_0}} & \phantom{e^{i \alpha_0}} & e^{i \alpha_{d-1}}
\end{array}
\right)\right|\alpha_0,...,\alpha_{d-1} \in \mathbb{R}\right\},\quad V=\frac{1}{\sqrt{d}}\sum^{d-1}_{j,k=0}
 \omega^{j k}\ket{j}\bra{k}, \omega= e^{i \frac{2 \pi}{d}}, H'=VHV^{-1}
\end{displaymath}


\begin{proposition}\label{pro1}
Both $H$ and $H'$ are closed connected subgroups of the compact Lie group of unitaries $G=U(d)$.
\end{proposition}

\begin{proof}
Let $S^1=\{e^{i\alpha}|\alpha\in \mathbb{R}\}$. Clearly, the circle $S^1$ is closed and connected. Since $H\cong
S^1\times \cdots\times S^1$, $H$ is also a closed connected group. It is obvious that $H'$ is topologically isomorphic
to $H$. Thus $H'$ is a closed connected group.
\end{proof}

We need two lemmas as follows.

\begin{lemma}\label{lem1}
\cite{Bry} Let $G$ be a compact Lie group. If $H_1,...,H_k$ are closed connected
 subgroups and they generate a dense group of $G$, then in fact they generate $G$.
\end{lemma}

\begin{lemma}\label{lem2}
\cite{Bry} Let $\mathfrak{h}= \textnormal{Lie } H$, $\mathfrak{h}'= \textnormal{Lie } H'$,
$\mathfrak{g}= \textnormal{Lie } U(d)$. If $\mathfrak{h}$ and $\mathfrak{h}'$ generate $\mathfrak{g}$ as a
Lie algebra, and $H$ and $H'$ are closed connected groups, then $H$ and $H'$ generate a dense subgroup of $U(d)$.
\end{lemma}

We choose the following matrices \cite{Yao} as a basis of the Lie algebra $\mathfrak{g}$:

\[\sigma_{x}^{(jk)} (0\leq j< k\leq d-1),  \sigma_{y}^{(jk)} (0\leq j< k\leq d-1),
\sigma_{z}^{(jk)} (j=0, 1\leq k \leq d-1), i I_d\]
where
\[\sigma_{x}^{(jk)}=i\ket{j}\bra{k}+i \ket{k}\bra{j}, \sigma_{y}^{(jk)}=\ket{j}\bra{k}-\ket{k}\bra{j},
\sigma_{z}^{(jk)}=i\ket{j}\bra{j}-i \ket{k}\bra{k}\].

It is easily checked that
\[
\mathfrak{h}=\left\{\left. \sum^{d-1}_{j=0} i \alpha_j \ket{j}\bra{j}\right |\alpha_j\in \mathbb{R}\right\}= span
\left\{ \sigma_z^{(0j)}(1\leq
j\leq d-1), iI_d\right\}
\]

\[
\mathfrak{h}'=\left\{ \left.\sum^{d-1}_{j=0} i \alpha_j V \ket{j}\bra{j}V^{-1}\right |\alpha_j\in \mathbb{R}\right\}= span
\left\{V\sigma_z^{(0j)}V^{-1}(1\leq j\leq d-1), iI_d\right\}
\]

\begin{theorem}\label{the1}
Let $\mathfrak{m}$ be the Lie subalgebra of $\mathfrak{g}$ generated by $\mathfrak{h}$ and $\mathfrak{h}'$.
Then all the $\sigma_x^{(jk)}(0\leq j<k\leq d-1)$ and $\sigma_y^{(jk)}(0\leq j<k\leq d-1)$ are included in $\mathfrak{m}$.
\end{theorem}

\begin{proof}
$\forall t \in \{ 1,..., d-1\}$, $V\sigma^{(0t)}_{z}V^{-1}\in \mathfrak{m}$.
\[
V\sigma^{(0t)}_{z}V^{-1}=\left(\frac{1}{\sqrt{d}} \sum^{d-1}_{j,k=0} \omega^{jk}\ket{j}\bra{k} (i\ket{0}\bra{0}-
i \ket{t}\bra{t})\left( \frac{1}{\sqrt{d}} \sum^{d-1}_{j_1,k_1=0} \bar{\omega}^{j_1k_1}\ket{k_1}\bra{j_1}
\right)
\right)=\frac{i}{d}\sum^{d-1}_{j,j_1=0}(1-\omega^{(j-j_1)t}\ket{j}\bra{j_1})
\]

Thus $\sum^{d-1}_{t=1}V\sigma^{(0t)}_{z}V^{-1}\in \mathfrak{m}$. By direct calculation,
\[\chi:= \sum^{d-1}_{t=1}V\sigma^{(0t)}_{z}V^{-1}=\sum^{d-1}_{t=1}\frac{i}{d}\left(\sum^{d-1}_{j,j_1=0}
(1-\omega^{(j-j_1)t}\ket{j}\bra{j_1}
\right)=\sum^{(d-1)}_{j,j_1=0,j\neq j_1}i\ket{j}\bra{j_1}=\sum_{0\leq j<k\leq d-1}\sigma_x^{(jk)}\in \mathfrak{m}
\]

We give the Lie products between $\sigma_x$, $\sigma_y$ and $\sigma_z$ as follows.
\begin{center}
\begin{displaymath}
 \left\{ \begin{array}{cccc}
  &[\sigma_x^{(0t)},\sigma_z^{(0t)}] &=& 2\sigma_y^{(0t)} \\
  &[\sigma_x^{(0k)},\sigma_z^{(0t)}] &=& \sigma_y^{(0k)}, k\neq t \\
  &[\sigma_x^{(tk)},\sigma_z^{(0t)}] &=& -\sigma_y^{(tk)}, 0<t\neq k \\
  &[\sigma_x^{(jt)},\sigma_z^{(0t)}] &=& \sigma_y^{(jk)}, 0<j\neq t
\end{array} \right.\quad (1)
\end{displaymath}

\begin{displaymath}
 \left\{ \begin{array}{cccc}
  &[\sigma_y^{(0k)},\sigma_z^{(0k)}] &=& -2\sigma_x^{(0k)} \\
  &[\sigma_y^{(0k)},\sigma_z^{(0t)}] &=& -\sigma_x^{(kt)}, k\neq t \\
  &[\sigma_y^{(jk)},\sigma_z^{(0j)}] &=& -\sigma_x^{(jk)}, 0<j\neq k \\
  &[\sigma_y^{(jk)},\sigma_z^{(0k)}] &=& \sigma_x^{(jk)}, 0<j\neq k
\end{array} \right.\quad (2)
\end{displaymath}
\end{center}

From the Lie products listed above, we have
\[[\chi,\sigma_z^{(0t)}]=\sum_{k=1}^{d-1}\sigma_y^{(0k)}+\sum_{k=0}^{d-1}\sigma_y^{(kt)}\in\mathfrak{m}, \forall
t \in \{1,...,d-1\}.
\]

\begin{displaymath}
 \forall u \in \{1,...,d-1\}, [\sum_{k=1}^{d-1}\sigma_y^{(0k)}+\sum_{k=0}^{d-1}\sigma_y^{(kt)}, \sigma_z^{(0u)}]
=\left\{ \begin{array}{cc}
 -\sigma_x^{0u}-\sum_{k=0,k\neq u}^{d-1}\sigma_x^{(ku)},& t\neq u\\
  -2\sigma_x^{(0u)}-2\sum_{k=0,k\neq u}^{d-1}\sigma_x^{(ku)},& t= u\\
\end{array} \right.\quad
\end{displaymath}

Therefore, $\sigma_{x}^{(0u)}+\sum_{k=0,k\neq u}^{d-1}\sigma_x^{(ku)}\in \mathfrak{m}, \forall
u \in \{1,...,d-1\}.$

Furthermore, for $u,v \in \{1,...,d-1\}$,
\begin{displaymath}
[\sigma_{x}^{(0u)}+\sum_{k=0,k\neq u}^{d-1}\sigma_x^{(ku)},\sigma_{z}^{(0v)}]=
\left\{ \begin{array}{cc}
2\sigma_y^{(0u)}-\sigma_y^{(uv)}\in \mathfrak{m}, & v\neq u\\
4\sigma_y^{(0u)}+\sum_{k=1,k\neq u}^{d-1}\sigma_y^{(ku)}\in \mathfrak{m}, & v= u
\end{array} \right.
\end{displaymath}

Thus \[\sum_{v=1,v\neq u}^{d-1}\left(2\sigma_y^{(0u)}-\sigma_y^{(vu)}\right)=2(d-2)\sigma_y^{(0u)}-
\sum_{v=1,v\neq u}^{d-1}\sigma_y^{(vu)}\in \mathfrak{m}\]
\[\left(2(d-2)\sigma_y^{(0u)}-
\sum_{v=1,v\neq u}^{d-1}\sigma_y^{(vu)}\right)+4 \sigma_y^{(0u)}+\sum_{k=1,k\neq u}^{d-1}\sigma_y^{(ku)}=2d
\sigma_y^{(0u)}\in \mathfrak{m}
\]

i.e., $\sigma_y^{(0u)}\in \mathfrak{m},\forall u\in \{1,...,d-1\}$.

Immediately, we get $\sigma_y^{(vu)}\in \mathfrak{m},\forall u\neq v, u,v \in \{1,...,d-1\}$.

Up to now, all the $\sigma_y^{(jk)}(0\leq j< k \leq d-1)$ are included in $\mathfrak{m}$. Still from the Lie
products listed in (\ref{equ2}), we know that all the $\sigma_x^{(jk)}(0\leq j< k \leq d-1)$ are included in $\mathfrak{m}$.
\end{proof}

Theorem (\ref{the1}) means that $\mathfrak{h}$ and $\mathfrak{h'}$ generate the Lie algebra $\mathfrak{g}$.
It follows from proposition(\ref{pro1}), lemma (\ref{lem1}) and lemma (\ref{lem2}) that $H$ and $H'$ generate
$U(d)$, which means qudit ZX Calculus contains all single qudit unitary transformations. Therefore the qudit ZX calculus is universal for quantum mechanics.

\section{Conclusion and Future Work}
In this paper, we introduce a dichromatic calculus (RG) for qutrit systems. We show that the decomposition of the qutrit Hadamard gate is non-unique and not derivable from the dichromatic calculus. As an application of the dichromatic calculus, we depict a quantum algorithm with a single qutrit. Furthermore, for any $d$, we prove that the qudit ZX calculus contains all single qudit unitary transformations. It follows that qudit ZX calculus, with qutrit dichromatic calculus as a special case, is universal for quantum mechanics.

There are many issues requiring further exploration. Here we just list a few of them as follows. First, does there exist a formula in which each unitary is decomposed into X and Z phase gates? Second, is the dichromatic ZX calculus complete for qutrit stabilizer quantum mechanics?  Finally, is the dichromatic ZX calculus incomplete for qutrit quantum mechanics?


\nocite{*}
\bibliographystyle{eptcs}
\bibliography{DiCal_arXiv}
\end{document}